\documentclass{article}

\usepackage[preprint]{neurips_2026}

\usepackage[utf8]{inputenc} 
\usepackage[T1]{fontenc}    
\usepackage{hyperref}       
\usepackage{url}            
\usepackage{booktabs}       
\usepackage{amsfonts}       
\usepackage{nicefrac}       
\usepackage{microtype}      
\usepackage{xcolor}         
\usepackage{edgeStd}

\usepackage{amsthm}
\usepackage{aliascnt}
\usepackage[capitalize,noabbrev]{cleveref}
\usepackage{algorithm}
\usepackage[noend]{algorithmic}  
\usepackage{wrapfig} 
\theoremstyle{plain}
\newtheorem{theorem}{Theorem}[section]
\newtheorem{proposition}[theorem]{Proposition}
\newtheorem{lemma}[theorem]{Lemma}

\theoremstyle{definition}

\newaliascnt{assumption}{theorem}
\newtheorem{assumption}[assumption]{Assumption}
\aliascntresetthe{assumption}
\theoremstyle{remark}

\crefname{assumption}{Assumption}{Assumptions}
\Crefname{assumption}{Assumption}{Assumptions}
\newcommand{\is}{\text{is}}
\newcommand{\negap}{\text{NE-gap}}
\newcommand{\tv}{\text{TV}}

\newcommand{\noclip}{\text{noclip}}
\allowdisplaybreaks
\usepackage{makecell}

\usepackage{enumitem}

\makeatletter
\providecommand{\captionof}[2]{\def\@captype{#1}\caption{#2}}
\makeatother
\setlist[itemize]{leftmargin=1em,topsep=1pt,itemsep=1pt,parsep=0pt}
\title{Decentralized Diffusion Policy Learning for Enhanced Exploration in Cooperative Multi-agent Reinforcement Learning}

\author{%
Yuyang Zhang$^{1,2}$
\quad Haldun Balim$^{1}$
\quad Na Li$^{1}$\\
$^1$SEAS, Harvard University \quad $^2$Kempner Institute, Harvard University\\
\texttt{\{yuyangzhang@g,
hbalim@fas, nali@seas\}.harvard.edu}
}

\begin{document}

\maketitle

\begin{abstract}
  Cooperative multi-agent reinforcement learning (MARL) involves complex agent interactions and requires effective exploration strategies. A prominent class of MARL algorithms, decentralized softmax policy gradient (DecSPG), addresses this through energy-based policy updates. In practice, however, such energy-based policies are intractable to maintain and are commonly projected onto the Gaussian policy class. In this work, we show that the limited expressiveness of Gaussian policies severely hinders exploration in DecSPG, and this limitation worsens as the number of agents grows. To address this issue, we propose decentralized diffusion policy learning (DDPL), which parameterizes each agent's policy with a denoising diffusion probabilistic model, an expressive generative model that captures multi-modal action distributions for enhanced exploration. DDPL enables efficient online training of diffusion policies via importance sampling score matching (ISSM), a novel training method with theoretical guarantee. We evaluate DDPL on representative continuous-action MARL benchmarks, including multi-agent particle environment, multi-agent MuJoCo, IsaacLab, and JAX-reimplemented StarCraft multi-agent challenge, and observe consistently improved performance.
\end{abstract}

\section{Introduction}
Multi-agent reinforcement learning (MARL) has found broad applicability across a range of domains \citep{marl_survey}, including traffic networks \citep{intro_1}, smart grids \citep{intro_2}, smart buildings \citep{intro_3}, and robotics \citep{intro_4,intro_5}. In MARL, agents learn policies by interacting with the environment and adapting to the evolving behavior of other agents, with the goal of maximizing their rewards. Compared to single agent reinforcement learning, MARL poses additional challenges for exploration in the high-dimensional joint action space.

In this paper, we focus on cooperative MARL, where agents share a common global reward and seek to maximize it collectively. In this setting, effective exploration is essential: agents must try diverse coordination patterns to discover high-reward joint actions, while premature commitment to a single pattern can trap the team in a suboptimal equilibrium. A prominent class of algorithms for this setting is decentralized softmax policy gradient (DecSPG) \citep{marl_theory1} and its variants \citep{marl_theory1_e2,marl_theory1_e3,marl_theory1_e4,marl_theory1_e5,marl_1,marl_3}
These methods iteratively perform decentralized energy-based policy updates, which naturally generate multimodal action distributions for diverse exploration and provably converge to Nash equilibria when policy updates are exact.

In practice, however, such energy-based action distributions are notoriously difficult to represent exactly \citep{energy_hard}. Most practical algorithms therefore project these energy-based policies onto simple policy classes, most commonly Gaussian policies \citep{marl_1,marl_3}, favoring computational efficiency at the expense of expressiveness.
Different from the original energy-based target policies, Gaussian policies lack expressiveness and often fail to capture multimodal action distributions. We find that this restriction can significantly hinder exploration and cause algorithms to converge prematurely to suboptimal equilibrium.
The failure mode is particularly severe in multi-agent settings, where the number of viable coordination patterns grows with the number of agents, and the cost of missing them compounds accordingly. 

This leads to a central challenge when applying the DecSPG-based algorithms in practice: designing policy representations that are both expressive enough for effective exploration and complex coordination and tractable enough for online training and sampling. 
Fortunately, diffusion models offer a natural remedy \citep{diffusion_1,diffusion_2}. They are expressive generative models capable of representing complex, multimodal distributions while remaining efficient to train and sample from. This presents a great opportunity to parameterize reinforcement learning policies using diffusion models, especially in cooperative MARL where exploration with expressive policies is critical. 

The main obstacle to deploying diffusion policies in DecSPG-based algorithms is the lack of efficient online training algorithms. Standard training algorithms, e.g., score matching and its variants \citep{scorematching_1,scorematching_2,scorematching_3}, are primarily designed for offline settings and typically require samples from the target distribution, which are rarely available in online MARL. Most existing work either proposes alternative training methods that are computationally heavy \citep{diffpolicy_4_heavy,diffpolicy_5_heavy,diffpolicy_6_heavy} or requires additional information, e.g., the gradient of the Q function, that is nontrivial to estimate \citep{diffpolicy_1_qgrad,diffpolicy_2_qgrad,diffpolicy_3_qgrad}. Drawing on ideas from loss reweighting \citep{diffpolicy_7_is,diffpolicy_8_is}, we derive an online variant of the score-matching loss that trains diffusion policies without requiring samples from the target policy distribution.


\textbf{Contributions.} This paper formally shows that the limited expressiveness of standard policy parameterizations is a key obstacle to effective exploration in cooperative MARL, and develops a diffusion-based approach to overcome this limitation. Specifically,
\begin{itemize}[itemsep=0pt, topsep=0pt, parsep=0pt, leftmargin=*]
    \item First (\Cref{sec:alg}), we identify an explicit failure mode of unimodal policy classes and motivate the need for expressive, multimodal policy representations. We construct a representative environment in which Gaussian approximations to the energy-based policy updates of DecSPG severely restrict exploration and lead to convergence to suboptimal policies. Under mild conditions, we prove that the success probability of the resulting algorithm is at most $0.45^n$, decaying exponentially to zero as the number of agents increases. 
    \item Then (\Cref{sec:is}), we propose decentralized diffusion policy learning (DDPL), a DecSPG-based algorithm that parameterizes decentralized agent policies with denoising diffusion probabilistic models (DDPMs). To enable online training, we derive an importance-sampling-based score matching objective that moves the diffusion policy toward the target energy-based distribution without requiring samples from that distribution. We also provide theoretical sample-complexity guarantees for the proposed training method.
\item Finally (\Cref{sec:sim}), we evaluate DDPL on representative continuous-action MARL benchmarks, including the multi-agent particle environment (MPE), multi-agent MuJoCo (MaMuJoCo), IsaacLab bi-shadow-hand, and JAX-reimplemented StarCraft Multi-Agent Challenge (SMAX) \citep{jaxmarl}, demonstrating improved performance and sample efficiency.
\end{itemize}

\textit{Notations.} For integer $n$, we use $[n]$ to denote $\{1, 2, \cdots, n\}$. Given agent set $[n]$ and integer $i$, we use $-i$ to denote agents $[n]\backslash \{i\}$. We use $\{0,1\}^n$ to denote all binary vectors in $\bbR^n$, i.e., $\{0,1\}^n=\{x=(x_1,\cdots, x_n)\in\bbR^n: x_i\in\{0,1\}\}$. 
We use $\bm1_n$ to denote the $n$-dimensional vector $(1, \cdots, 1)$. 
We use $\lesssim$ to hide dependencies on absolute constants and log factors. 

\section{Preliminaries}
\subsection{Cooperative Multi-Agent Reinforcement Learning}
In cooperative MARL, agents collectively maximize a global reward. A standard model for this setting is the identical-interest Markov game with global state and decentralized policies \citep{setting}. The formal definition is as follows.

\textbf{Identical-Interest Markov Games.}
An identical-interest Markov game is defined by the tuple $\calM=(\calS, \rho^0, \{\calA_i\}_{i=1}^n, r, P, \gamma)$.
Here $\calS$ is the global state space and $\rho^0\in\Delta(\calS)$ is the initial state distribution. Each agent $i\in[n]$ has an individual action space $\calA_i$. We write the joint action as $\bma\coloneqq(\bma_1,\cdots,\bma_n)\in\calA_1\times\cdots\times\calA_n \coloneqq \calA$. The reward function $r:\calS\times\calA\mapsto [r_{\min},r_{\max}]$ assigns a common reward to all agents, and the transition kernel $P:\calS\times\calA\mapsto \Delta(\calS)$ specifies the next-state distribution. The scalar $\gamma\in(0,1)$ is the discount factor.

A policy $\pi:\calS\mapsto\Delta(\calA)$ is called decentralized if it factorizes across agents as $\pi(\bma|\bms)=\prod_{i=1}^n \pi_i(\bma_i|\bms)$. We denote such a policy by $\pi=\pi_1\times\cdots\times\pi_n$.
For a decentralized policy $\pi$, we define the action-value function, value function, and objective as $Q^{\pi}(\bms,\bma)\coloneqq \bbE_{\pi}\bs{\sum_{t=0}^{\infty} \gamma^t\b{r(\bms^t,\bma^t)}|\bms^0=\bms,\bma^0=\bma},\quad V^{\pi}(\bms) \coloneqq \bbE_{\bma\sim\pi\b{\cdot|\bms}}\bs{Q^{\pi}(\bms,\bma)}$, and $J\b{\pi} \coloneqq \bbE_{\bms\sim\rho^0}\bs{V^{\pi}(\bms)}$, respectively.
For each agent $i$, we further define the averaged action-value function, which evaluates the effect of agent $i$'s local action $\bma_i$ averaged over the actions of all other agents $\bma_{-i}$ under the decentralized policy $\pi$:
\[\bar{Q}_i^{\pi}(\bms,\bma_i)\coloneqq\bbE_{\bma_{-i}\sim\pi_{-i}(\cdot|\bms)}\left[Q^{\pi}(\bms,(\bma_i,\bma_{-i}))\right].\]

\textbf{Nash Equilibrium (NE).}
A policy $\pi^\star=\pi^\star_1\times\cdots\times\pi^\star_n$ is a Nash equilibrium if $J(\pi_i^\star, \pi_{-i}^\star) \geq J(\pi_i', \pi_{-i}^\star)$ for all $i\in[n], \pi_i'$.
Furthermore, we define the Nash Equilibrium Gap (\negap) of a policy $\pi=\pi_1\times\cdots\times\pi_n$ as $\negap(\pi) \coloneqq \max_{i\in[n]} \negap_i(\pi)$, where $\negap_i(\pi) \coloneqq\max_{\pi_i'} J(\pi_i', \pi_{-i}) - J(\pi_i,\pi_{-i})$.
Policy $\pi$ is called an $\epsilon$-NE if $\negap(\pi)\leq \epsilon$. 

\subsection{Denoising Diffusion Probabilistic Model}
The denoising diffusion probabilistic model (DDPM) is an expressive generative model for learning complex distributions. 
\textit{The forward process} corrupts the data from the data distribution $\bma_{0}\sim p_0$ towards the noise distribution $\calN(0,I)$ by adding Gaussian noise, i.e., $\bma_{\tau} = \sqrt{1-\beta_\tau} \bma_{\tau-1} + \sqrt{\beta_\tau} \epsilon$.
Here $\{\beta_\tau\}_{\tau=1}^H$ is the variance schedule and $\epsilon\sim\calN(0,I)$. With $\alpha_\tau = 1 - \beta_\tau$, $\bar{\alpha}_\tau = \prod_{h=1}^\tau \alpha_h$,
we can express $\bma_{\tau}$ as $\bma_{\tau} = \sqrt{\bar{\alpha}_\tau} \bma_{0} + \sqrt{1-\bar{\alpha}_\tau}\epsilon$. We denote the marginal distribution of $\bma_{\tau}$ as $p_{\tau}$ and the conditional distribution of $\bma_{\tau}|\bma_{0}$ as $p_{\tau|0}$.

\textit{The reverse process} denoises samples from $\bma_H\sim \calN(0,I)$ towards data distribution $p_0$: $\bma_{\tau-1} = \frac{1}{\sqrt{\alpha_{\tau}}}\b{\bma_{\tau}+\beta_{\tau} s_{\tau}^{\theta}(\bma_{\tau})} + \sigma_\tau\epsilon$.
Here $\sigma_{\tau}$ is a hyperparameter and we let  $\sigma_{\tau}^2=\beta_{\tau}/\alpha_{\tau}$ in this paper.
$s_{\tau}^{\theta}(\bma_{\tau})$ is the score network with parameters $\theta$, which can be trained via denoising score matching \citep{diffusion_1} by minimizing $\calL^{\text{marginal}}(\theta) = \sum_{\tau=1}^H \bbE_{\bma_{\tau}\sim p_{\tau}} \frac{\beta_{\tau}^2}{\sigma_{\tau}^2\alpha_{\tau}}\norm{s_{\tau}^{\theta}\b{\bma_{\tau}} - \nabla_{\bma_{\tau}}\log p_{\tau}(\bma_{\tau})}^2$.
It is shown in \citet{rldiff_1} that minimizing $\calL^{\text{marginal}}$ is equivalent to minimizing the following conditional loss that is more tractable:
\begin{equation}\begin{split}\label{eq:loss_naive}
    \calL(\theta)
    = {}& \bbE_{\bma_{0}\sim p_0}\sum_{\tau=1}^H \frac{\beta_{\tau}^2}{\sigma^2_{\tau}\alpha_{\tau}}\bbE_{\epsilon_{\tau}\sim\calN(0,I)} \norm{s_{\tau}^{\theta}\b{\sqrt{\bar{\alpha}_{\tau}}\bma_{0}+\sqrt{1-\bar{\alpha}_{\tau}}\epsilon_{\tau}} - \frac{\epsilon_{\tau}}{\sqrt{1 - \bar{\alpha}_\tau}}}^2.
\end{split}\end{equation}

\section{Comparing Multi-Modal Policies with Gaussian Policies in MARL}\label{sec:alg}
We first present decentralized softmax policy gradient (DecSPG) \citep{marl_theory1},
which maintains decentralized energy-based policies and converges to NEs under exact policy updates. In practice, the exact updates are intractable when action space is large or even continuous, and the policies are commonly projected onto the Gaussian policy class. In this section, we investigate how this projection affects performance through concrete examples. We show that Gaussian policies cause severe under-exploration, and this failure mode amplifies exponentially with the number of agents.

\subsection{Decentralized Softmax Policy Gradient}
\Cref{alg:exact} iteratively estimates each agent's averaged $Q$ function $\wh{Q}_i^{\pi^{k-1}}$ and updates each agent's policy towards the energy-based target $\pi^{\star,k}_i \propto \pi_i^{k-1}\exp(\frac{\eta}{1-\gamma}\wh{Q}_i^{\pi^{k-1}})$. Each agent thereby captures the influence of other agents' policies $\pi_{-i}^{k-1}$ through the averaged $Q$ function and gradually reweights toward actions with higher rewards under $\pi_{-i}^{k-1}$. This energy-based target is generally intractable, and the standard approach is to project it onto a parametric class $\Pi_i$.

\begin{algorithm}[h]
    \caption{Decentralized Softmax Policy Gradient (Dec-SPG) \citep{marl_theory1}}
    \label{alg:exact}
    \begin{algorithmic}[1]     
        \STATE \textbf{Init: } Initial policy $\pi^0=\pi_1^0\times\cdots\times\pi_n^0$, learning rate $\eta$, batch size $B$;
        \FOR{epoch $k\in[K]$}
        \STATE Take $B$ steps with policy $\pi^{k-1}$ and store to buffer $\calD=\{\bms,\bma,r,\bms'\}$;
        \STATE Estimate average $Q$ function $\bar{Q}_i^{\pi^{k-1}}$ for each agent $i$ with data in $\calD$, denoted by $\wh{Q}_i^{\pi^{k-1}}$;
        \STATE Update policy $\pi_i^k(\cdot|\bms)$ for each agent $i$:
        \begin{equation*}\begin{split}
            \pi_i^k(\cdot|\bms) \gets \arg\min_{\pi_i\in\Pi_i} \text{KL}\b{\pi_i\|\pi_i^{\star,k}(\cdot|\bms)}, \quad
            \pi_i^{\star,k}(\cdot|\bms) \propto \pi_i^{k-1}(\cdot|\bms)\exp\b{\tfrac{\eta}{1-\gamma}\wh{Q}_i^{\pi^{k-1}}(\bms,\cdot)}.
        \end{split}\end{equation*}
        \ENDFOR
    \end{algorithmic}
\end{algorithm}

When the $Q$ estimation and policy updates are exact, \Cref{alg:exact} provably converges to a Nash equilibrium given finite state and action spaces (\Cref{thm:marl}). The average squared NE-gap decays at rate $\calO(n/K)$, where $K$ and $n$ are the number of epochs and agents, respectively. This guarantee, however, requires the policy update to be exact, i.e., $\pi_i^k = \pi_i^{\star,k}$, which is rarely satisfied in practice due to the policy projection onto $\Pi_i$.


\subsection{Multi-Modal Policies Improve Exploration}\label{subsc:example}
We now investigate how the policy class $\Pi_i$ affects the performance of \Cref{alg:exact}, comparing all Gaussian policies against all policies. Although Gaussians are popular for their efficiency, we show that their unimodal nature causes insufficient exploration, breaking the guarantee of \Cref{alg:exact}. We first build intuition with a single-agent MDP, then extend to an identical-interest game where the failure amplifies exponentially with the number of agents.

\subsubsection{A Single-Agent Example}
We construct MDP $\calM_{\text{sa}}$ with three states $\calS=\{\bms_0, \bms_1, \bms_2\}$, 
fixed initial state $\bms_0$,
action space $\calA_1=\bbR$, and discount factor $\gamma=0.1$. The reward function $r(\bms,\bma)$ is defined as:
\begin{equation*}\begin{split}
    r(\bms, \bma) = \mathbb{I}\b{\bms = \bms_0}r_{\text{start}}(\bma) + \mathbb{I}\b{\bms=\bms_1}c_1 + \mathbb{I}\b{\bms=\bms_2}0
\end{split}\end{equation*}
with $r_{\text{start}}(\bma)=-\nu\bma^2 + c_0\exp(-\alpha(\bma+\beta)^2) + c_0\exp(-\alpha(\bma-\beta)^2)$ and parameters $\nu=10, \alpha=250, \beta=0.5, c_0=2.7$, and $c_1=20000$. The transition probability is defined as:
\begin{align*}
    \bbP(\bms_1|\bms_0,\bma) = p_{\text{start}}(\bma), \quad \bbP(\bms_2|\bms_0,\bma) = 1-\bbP(\bms_1|\bms_0,\bma), \quad \bbP\b{\bms_2|\bms_1, \bma} = \bbP\b{\bms_2|\bms_2, \bma} = 1.
\end{align*}
where $p_{\text{start}}(\bma)=c_p\exp(-\alpha(\bma+\beta)^2)\cdot \mathbb{I}\b{|\bma+\beta|\leq 1/4}$ with $c_p=0.005$. 
Intuitively, this MDP rewards actions near $\pm\beta$ and grants a small probability of transitioning to `goal' state $\bms_1$ for actions near $-\beta$. Otherwise the environment transitions to the absorbing state $\bms_2$ with zero reward. Learning an optimal policy requires extensive exploration in $\bms_0$.

\begin{figure*}[h]\begin{center}
       \centering
       \includegraphics[width=0.9\linewidth]{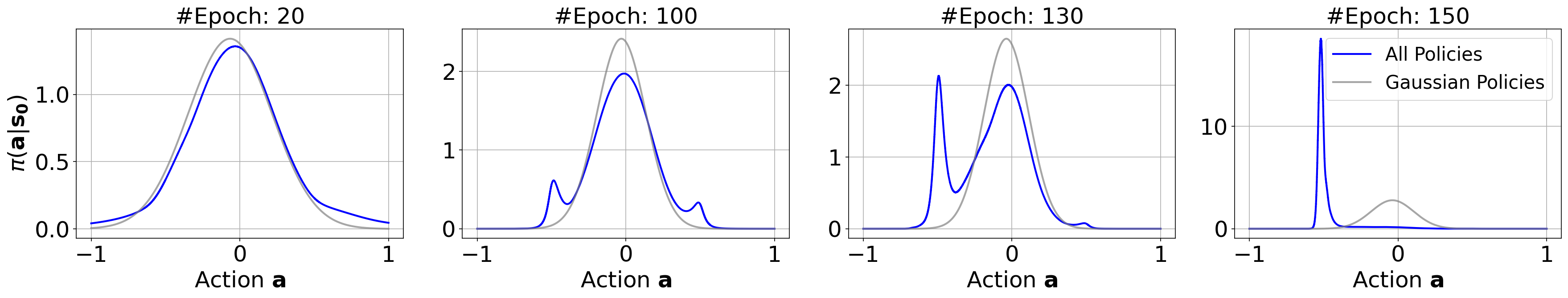}\\[-0pt]
    \caption{Policies during training with Gaussian policies and all policies in the single-agent example.}
    \label{fig:sa_policies}
\end{center}\end{figure*}
We then run \Cref{alg:exact} with batch size $B=200$, $\eta=0.01(1-\gamma)$, and $\Pi_i$ being either Gaussian policies or all policies. We plot representative learned policies (\Cref{fig:sa_policies}) and the average reward of the policies (\Cref{fig:sa_value}, left) during training. Looking at the trained policies (\Cref{fig:sa_policies}), we observe:
\begin{itemize}
    \item Initially (Epoch 20), both policy classes lack action samples near $\pm\beta$, and $\wh{Q}(\bms_0,\bma)\approx-\nu \bma^2$. This leads to a zero-mean Gaussian policy proportional to $\exp\b{-\eta/(1-\gamma)\cdot k\nu\bma^2}$ in epoch $k$.
    \item Later (Epoch 100), more samples near $\pm\beta$ accumulate and $\wh{Q}(\bms_0,\bma)$ captures reward terms $c_0\exp\b{-\alpha(\bma\pm\beta)^2}$. Therefore, \textit{with all policies}, the energy-based reweighting term $\exp(\eta\wh{Q}/(1-\gamma))$ gradually increases action probability near $\bma=\pm\beta$, resulting in a multi-modal exploration distribution. In contrast, \textit{with Gaussian policies}, the updated policy is stuck near $0$. The reward signal near $\bma=\pm\beta$ in $\exp(\eta\wh{Q}/(1-\gamma))$ is too weak to take effect in the Gaussian policy projection. Therefore, the policy continues to shrink its variance due to the $-\nu\bma^2$ term in $\wh{Q}$, leading to insufficient exploration near $\pm\beta$. 
    \item Finally (Epoch 130 \& 150), with more samples near $-\beta$, \textit{\Cref{alg:exact} with all policies} reaches state $\bms_1$ and converges to optimal action $\bma=-\beta$ in state $\bms_0$. However, \textit{\Cref{alg:exact} with Gaussian policies} continues to output policy centered at $0$.
\end{itemize}

\begin{wrapfigure}{r}{0.5\textwidth}
    \centering
    \begin{minipage}[t]{0.48\linewidth}
       \centering
       \includegraphics[width=\linewidth]{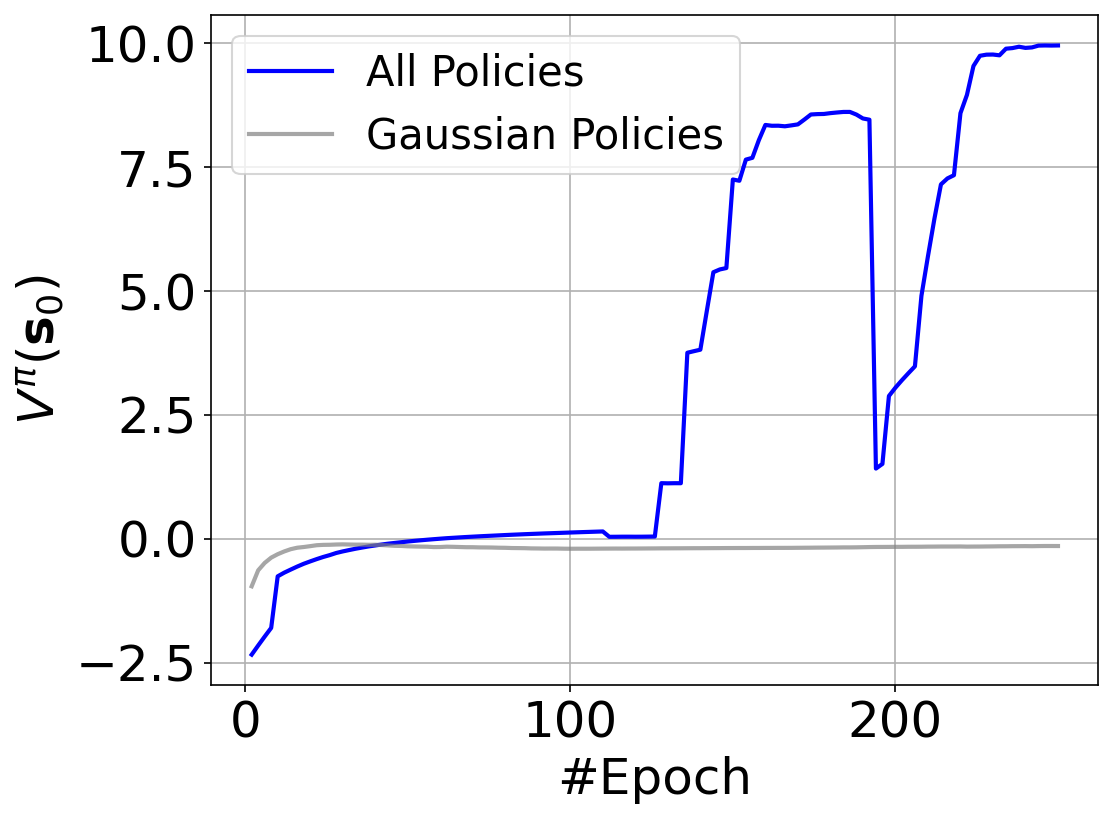}
    \end{minipage}
    \begin{minipage}[t]{0.48\linewidth}
       \centering
       \includegraphics[width=\linewidth]{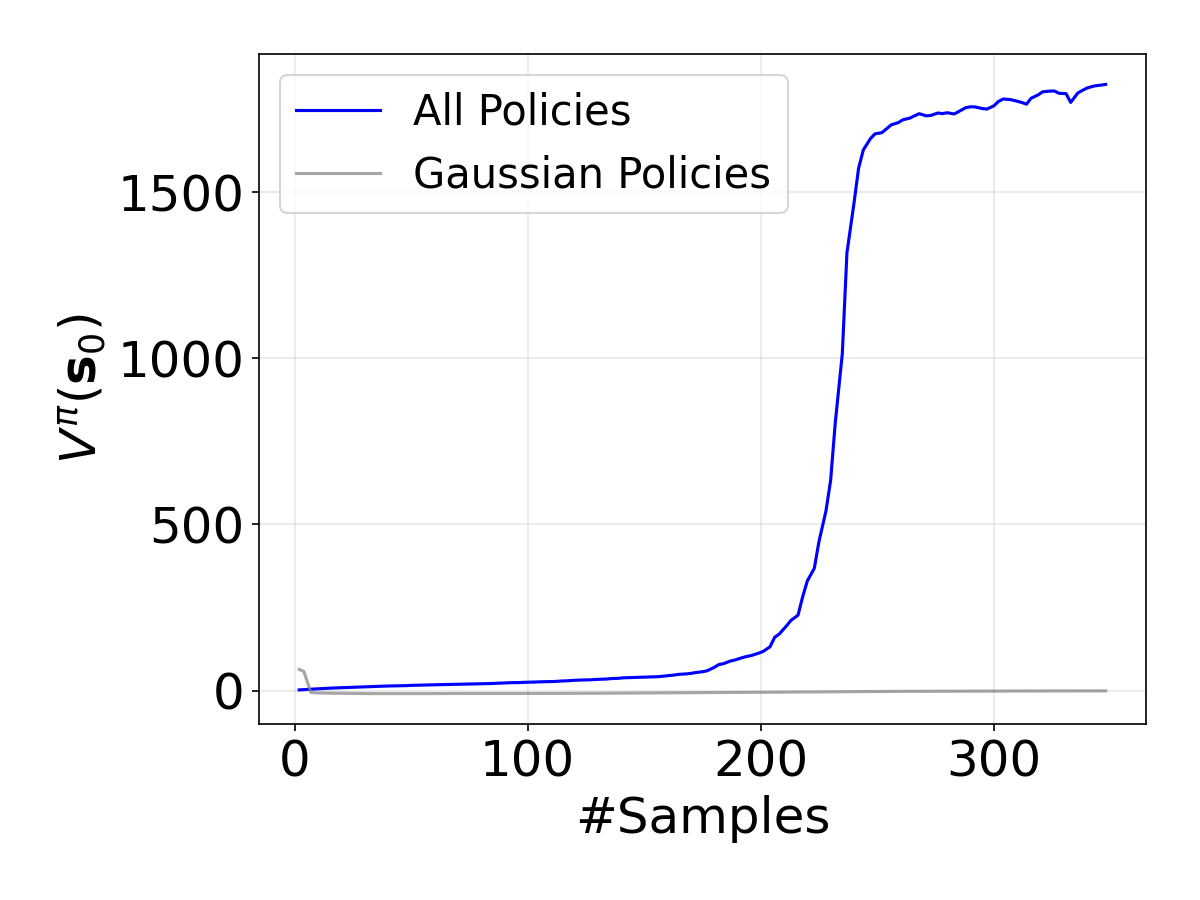}
    \end{minipage}
    \caption{Expected return during training for \Cref{alg:exact} with all policies and Gaussian policies. Left: single-agent. Right: multi-agent with $n=2$.}
    \label{fig:sa_value}
    \vspace{-2pt}
\end{wrapfigure}
The following theorem confirms this failure mode theoretically. It shows that starting from a zero-mean Gaussian with modest variance, similar to the first plot in \Cref{fig:sa_policies}, the Gaussian-projected policy update keeps the zero mean and monotonically shrinks the variance. The policy converges to a delta distribution around $0$, leading to a low probability of reaching $\bms_1$ throughout the entire algorithm.
The assumption that $\htQ_1^{\pi}$ is accurate for reward $r^\neg_{\text{sa}}$ isolates the failure mode of Gaussian policies by only considering the estimation error at state $\bms_1$. This is a natural condition before $\bms_1$ is reached. In practice, other estimation errors compound the issue.

\begin{theorem}\label{lem:sa_failure}
    Consider \Cref{alg:exact} with batch size $B = 32$, learning rate $\eta=0.01(1-\gamma)$, and $\Pi_1$ being Gaussian policy class. Consider any epoch $k$ before $\bms_1$ is first reached and any policy $\pi_1^k=\calN(\mu^k,(\sigma^k)^2)$ with $\mu^k=0$. Suppose $\htQ_1^{\pi^k_1}$ is accurate for $\calM_{\text{sa}}$ with reward $r^\neg_{\text{sa}}(\bms, \bma) = \mathbb{I}\b{\bms = \bms_0}r_{\text{start}}(\bma)$. Performing policy update (line 5, \Cref{alg:exact}) with gradient descent yields a Gaussian policy $\pi^{k+1}_1=\calN(\mu^{k+1},(\sigma^{k+1})^2)$ with
    \begin{equation*}\begin{split}
        \mu^{k+1} = 0, \quad \frac{1}{\b{\sigma^{k+1}}^2} \geq \frac{1}{\b{\sigma^k}^2} + 0.128.
    \end{split}\end{equation*}
    Moreover, starting from $\pi_1^0=\calN(0,0.1)$, the probability of reaching state $\bms_1$ is at most $0.45$.\qed
\end{theorem}

\subsubsection{A Multi-Agent Example}

We now extend the above example to an identical-interest Markov game $\calM_{\text{ma}}$ with agents $[n]$, state space $\calS=\{\bm{-1}_n\}\cup\{0,1\}^n$, fixed initial state $\bms_{\text{start}}=\bm{-1}_n$, and individual action space $\calA_i=\bbR$. The reward function is defined as:
\begin{equation*}\begin{split}
    r(\bms, \bma) = \mathbb{I}\b{\bms = -\bm{1}_n}\sum_{i=1}^n r_{\text{start}}(\bma_i) + \mathbb{I}\b{\bms\neq\pm\bm{1}_n}\,c_1\sum_{i=1}^n[\bms]_i + \mathbb{I}\b{\bms=\bm{1}_n}\b{c_2 + c_3\exp\b{-\bma\t\bma}},
\end{split}\end{equation*}
with parameters $c_1=20000, c_2=30000$, and $c_3=10000$. 
The transition probability is defined as:
\begin{equation*}\begin{split}
    {}& \bbP\b{[\bms]_i=1|\bms=\bm{-1}_n,\bma} = p_{\text{start}}(\bma_i),\quad \bbP\b{[\bms]_i=0|\bms=\bm{-1}_n,\bma} = 1-p_{\text{start}}(\bma_i),\\
    {}& \bbP\b{\bm{0}|\bms,\bma} = 1, \quad \forall\bms\neq\bm{-1}_n, \forall \bma.
\end{split}\end{equation*}
$\calM_{\text{ma}}$ embeds $n$ copies of the single-agent MDP at the initial state $\bm{-1}_n$. Each agent $i$ independently transitions coordinate $[\bms]_i$ to $1$ with probability $p_{\text{start}}(\bma_i)$. The `goal' state $\bms=\bm{1}_n$ is reached only when all agents succeed. Therefore, each agent learns independently before $\bms=\bm{1}_n$ is reached, which we will show amplifies the single-agent failure mode exponentially in $n$.

\begin{proposition}\label{lem:ma_failure}
    Consider \Cref{alg:exact} with batch size $B = 32$, learning rate $\eta=0.01(1-\gamma)$, and $\Pi_i$ being the Gaussian policy class. For every $i\in[n]$ and every epoch $k$ where $[\bms]_i=1$ has not been reached, suppose $\htQ_i^{\pi^k}$ is accurate for $\calM_{\text{ma}}$ with reward $r_{\text{ma}}^\neg(\bms, \bma) = \mathbb{I}\b{\bms = -\bm{1}_n}\sum_{j=1}^n r_{\text{start}}(\bma_j) + \mathbb{I}\b{\bms\neq -\bm{1}_n}\,c_1\sum_{j\neq i}[\bms]_j$. 
    Then starting from $\{\pi_i^0=\calN(0,0.1)\}_{i\in[n]}$, the probability of reaching state $\bms=\bm{1}_n$ using \Cref{alg:exact} and Gaussian policies is at most $0.45^n$.\qed
\end{proposition}
Here, assuming $\htQ_i^{\pi^k}$ being accurate for reward $r_{\text{ma}}^\neg$ isolates the failure mode by only considering the estimation error in reward terms $c_1[\bms]_i$ and $(c_2 + c_3\exp\b{-\bma\t\bma})$, which is a natural condition before $[\bms]_i$ is reached. In practice, other estimation errors compound the issue. Under this assumption, the success probability of reaching $\bms=\bm1_n$ decays exponentially in the number of agents, indicating that Gaussian policies severely hurt the performance of the algorithm in the multi-agent setting. We confirm this empirically with $n=2$ in \Cref{fig:sa_value} (right): \Cref{alg:exact} with Gaussian policies stays near zero return throughout training, while \Cref{alg:exact} with all policies discovers the high-reward equilibrium and converges to near-optimal performance.
\section{Decentralized Diffusion Policy Learning}\label{sec:is}
To perform the energy-based policy updates in \Cref{alg:exact} efficiently and accurately, we parameterize each agent's policy using a DDPM, an expressive generative model capable of representing complex multi-modal distributions. We then develop a provably efficient online training algorithm for diffusion models based on score matching, and present decentralized diffusion policy learning (DDPL), an online MARL algorithm with decentralized diffusion policies.
Without loss of generality, we consider bounded action space $\calA_i\subseteq\{\bma\in\bbR^{d}:\norm{\bma}\leq c_b\}$ for each agent $i$ throughout this section. Many popular MARL benchmarks, including MPE, MaMuJoCo, IsaacLab, SMAX, and real-world applications involve such bounded action spaces. 

\subsection{Importance Sampling Score Matching (ISSM)}
\textbf{Algorithm Design.} Consider policy $\pi_i^{\theta_i}$ parameterized by a DDPM with score network $s_{i,\tau}^{\theta_i}$ and parameter $\theta_i\in\Theta$. 
Given state $\bms$, sampling from $\bma_i\sim\pi_i^{\theta_i}(\cdot|\bms)$ is a two-step procedure. First, we sample $\bma_{i,0}$ via the $H$-step reverse process $\bma_{i,\tau-1} = \frac{1}{\sqrt{\alpha_{\tau}}}(\bma_{i,\tau}+\beta_{\tau} s_{i,\tau}^{\theta_i}(\bma_{i,\tau}|\bms)) + \sigma_\tau \epsilon$ and the noise schedule in \Cref{eq:scheduler}.
Second, we project $\bma_{i,0}$ onto $\calA_i$: $\bma_i = \text{Proj}_{\calA_i}(\bma_{i,0}) = \arg\min_{\bma\in\calA_i}\norm{\bma-\bma_{i,0}}$.

In epoch $k$, we aim to train $s_{i,\tau}^{\theta_i}(\cdot|\bms)$ so that distribution $\pi_i^k(\cdot|\bms) = \pi_i^{\theta_i}(\cdot|\bms)$ satisfies $\pi_i^k(\cdot|\bms) \propto \pi_i^{k-1}(\cdot|\bms)\exp(\frac{\eta}{1-\gamma}\wh{Q}_i^{\pi^{k-1}}(\bms,\cdot))$ for every state $\bms\in\calS$.
For simplicity, we consider the score network training for a specific agent $i$ and state $\bms\in\calS$.
A core challenge in applying standard score matching is that minimizing the denoising score matching loss $\calL_{i,\bms}(\theta_i)$ (\Cref{eq:loss_naive}) requires samples from the target distribution $\bma_i\sim\pi_i^{\star,k}(\cdot|\bms)$. Such samples are unavailable before the policy update.

To tackle this issue, we adopt the idea of \textit{importance sampling} by drawing samples from an easy-to-sample distribution $\tilp(\cdot|\bms)$, which will be specified later. 
With samples $\{\bma_i^{(j)}\sim\tilp(\cdot|\bms)\}_{j=1}^N$, we compute the following importance sampling score matching loss:
\begingroup\makeatletter\def\f@size{8}\check@mathfonts\makeatother
\begin{equation}\begin{split}\label{eq:flow_matching_is}
    {}& \calL_{i,\bms}^{\is}(\theta_i) \coloneqq \frac{1}{N}\sum_{j=1}^N \frac{\pi_i^{\star,k}(\bma_i^{(j)}|\bms)}{\tilp(\bma_i^{(j)}|\bms)} \cdot \sum_{\tau=1}^H \frac{\beta_{\tau}^2}{\sigma_{\tau}^2\alpha_{\tau}}\bbE_{\epsilon_{\tau}\sim \calN(0,I)}\norm{s_{i,\tau}^{\theta_i}\b{\left.\sqrt{\bar{\alpha}_{\tau}}\bma_i^{(j)}+\sqrt{1-\bar{\alpha}_{\tau}}\epsilon_{\tau}\right|\bms} - \frac{\epsilon_{\tau}}{\sqrt{1 - \bar{\alpha}_\tau}}}^2.
\end{split}\end{equation}
\endgroup
{\setlength{\lineskiplimit}{-4pt}%
In expectation, $\calL_{i,\bms}^{\is}(\theta_i)$ is equivalent to the original score matching loss $\calL_{i,\bms}(\theta_i)$ (\Cref{sec:issmloss}).
As a result, one only needs to minimize the $\calL_{i,\bms}^{\is}(\theta_i)$ within $\Theta$ and get minimizer $\wh{\theta}_i\coloneqq \mathop{\arg\min}_{\theta_i\in\Theta} \calL_{i,\bms}^{\is}(\theta_i)$.
Actions $\bma_i\sim\pi_i^{\wh{\theta}_i}(\cdot|\bms)$ can then be sampled using the reverse process with score network $s_{i,\tau}^{\wh{\theta}_i}(\cdot|\bms)$, followed by a projection onto action space $\calA_i$. 
\par}

\textbf{Theoretical Guarantee.} We now show that the importance sampling score matching is provably efficient. 
Due to space limit, we defer \Cref{assmp:score_comp} to the appendix, which poses standard and mild assumptions on the expressiveness, complexity, and Lipschitzness of score function class $\calS_{\Theta}$. 
\begin{theorem}\label{thm:isflow}
    Consider score function class $\calS_{\Theta}$ satisfying \Cref{assmp:score_comp} with constants $c_1$ and $c_2$. Consider target distribution $\pi_i^{\star,k}(\cdot|\bms)$ and sampling distribution $\tilp(\cdot|\bms)$, both supported on $\calA_i$. Let $c_w\coloneqq \max_{\bma}w(\bma)$, where $w(\bma) \coloneqq \pi^{\star,k}_i(\bma|\bms)/\tilp(\bma|\bms)$. 
    If sample size $N \gtrsim \log(d\sqrt{N}\log H/\delta)$, then $\pi_i^{\wh{\theta}_i}(\cdot|\bms)$ satisfies the following with probability at least $1-\delta$
    \begin{equation*}\begin{split}
        D_{\tv}^2\b{\pi_i^{\star,k}(\cdot|\bms)\|\pi_i^{\wh{\theta}_i}(\cdot|\bms)}\lesssim {}& \frac{d^2\log^6H}{H^2} + \frac{d\exp\b{\frac{3}{4}D_4\b{\pi_i^{\star,k}(\cdot|\bms)\|\tilp(\cdot|\bms)}}}{\sqrt{N}}.
    \end{split}\end{equation*}
    Here $d$ is action dimension, $H$ is diffusion denoising steps, and $D_4(p\|q)=\frac{1}{3}\ln\bbE_{x\sim q}\bs{(p(x)/q(x))^4}$ is the fourth-order Renyi divergence. $\lesssim$ and $\gtrsim$ absorbs constants $c_1, c_2, c_b, c_w$, and $\log(H, N, 1/\delta)$.
    \qed
\end{theorem}
The above theorem bounds error of the learned policy with two terms. The first term is the DDPM discretization error which decays with the number of denoising steps $H$. The second term is the statistical error which is proportional to $1/\sqrt{N}$ where $N$ is the sample size from sampling distribution $\tilp(\cdot|\bms)$.
We can therefore guarantee the accuracy of the learned target distribution $\pi_i^{\wh{\theta}_i}(\cdot|\bms)$ with large denoising steps $H$ and number of samples $N$. 

The one-step policy update error in \Cref{thm:isflow} also decays with $D_4(\pi_i^{\star,k}(\cdot|\bms)\|\tilp(\cdot|\bms))$. The closer $\tilp(\cdot|\bms)$ is to $\pi_{i}^{\star,k}(\cdot|\bms)$, the fewer samples we need for an accurate update. 
This motivates choosing $\pi_i^{k-1}(\cdot|\bms)$ as the sampling distribution. Theoretical guarantee of \Cref{alg:exact} (\Cref{thm:marl}) suggests a decreasing distance between the current policy $\pi_i^{k-1}(\cdot|\bms)$ and the target policy $\pi_i^{\star,k}(\cdot|\bms)$, since the sequence of policies $\{\pi_i^{\star,k}(\cdot|\bms)\}$ converges. Therefore, with $\tilp(\cdot|\bms)=\pi_i^{k-1}(\cdot|\bms)$, we can expect accurate policy update with less and less samples during training.
With this choice, the importance sampling weight in \Cref{eq:flow_matching_is} simplifies to ${\pi_i^{\star,k}(\bma_i^{(j)}|\bms)}/{\tilp(\bma_i^{(j)}|\bms)}=\exp(\frac{\eta}{1-\gamma}\wh{Q}_i^{\pi^{k-1}}(\bms,\bma_i^{(j)}))$.

We further quantify how the one-step policy update error propagates through epochs and influences NE-gap of the learned policies. We develop \Cref{thm:isflow_negap} building on previous literature \Cref{thm:marl}. Specifically, if diffusion updates are $\epsilon_{\is}$-accurate in TV distance, the squared average NE-gap, $\text{NE-gap}^2(\pi^k)/K$, scales as $\calO(n/K + n^2\epsilon_{\is})$ where $n$ is the number of agents and $K$ is the total number of episodes. Compared to the original rate in \Cref{thm:marl}, \Cref{thm:isflow_negap} only incurs an additive error linear in $\epsilon_{\is}$. Since NE-gap measures the deviation of a policy from a Nash equilibrium, this bound implies that the learned policies are approximate Nash equilibria as $K$ grows and $\epsilon_{\is}$ shrinks. Due to space constraints, we defer the formal statement and its proof to Appendix \ref{sec:proof-end-to-end}.

\subsection{The Full Decentralized Diffusion Policy Learning Algorithm}
The previous subsection establishes per-state updates for \Cref{alg:exact} with decentralized diffusion policies. To obtain a more practical algorithm, we extend \Cref{alg:exact} to large state spaces with neural network function approximation and replace exact ISSM loss minimizations with gradient steps, yielding \Cref{alg:extended}. In epoch $k$, each agent $i$ maintains a $Q$-network $\wh{Q}_i^{\psi_i^k}(\bms,\bma_i)$ and a diffusion policy $\pi_i^{\theta_i^k}(\bma_i|\bms)$ with score network $s_{i,\tau}^{\theta_i^k}(\bma_i|\bms)$, both taking state $\bms$ as input to generalize across the state space. The agent first estimates the $Q$ function (line 5), then updates the policy via one gradient step on the importance sampling score matching loss with $\pi_i^{\theta_i^{k-1}}$ as sampling distribution (line 6). Both updates are decentralized because they only require local actions $\bma_i$.


\begin{algorithm}[!h]
    \caption{Decentralized Diffusion Policy Learning (DDPL)}
    \label{alg:extended}
    \begin{algorithmic}[1]     
        \STATE \textbf{Init:} Policy parameters $\theta^0=\{\theta_i^0\}_{i\in[n]}$, Q parameters $\psi^0=\{\psi_i^0\}_{i\in[n]}$, learning rate $\eta$, rollout batch size $B_r$, update batch size $B_u$, smoothing factor $\xi$.
        \FOR{epoch $k\in[K]$}
        \STATE Rollout the current policy $\pi^{\theta^{k-1}}$ for $B_r$ steps and store to buffer $\calD$;
        \STATE For each agent $i$, sample batch $\calB_i\gets\{(\bms,\bma_i,r,\bms')\}$ of size $B_u$ from $\calD$ and next actions $\{\bma_{\bms',i}^{(j)}\sim\pi_i^{\theta_i^{k-1}}\}_{j\in[N]}$ for all $\bms'\in\calB_i$;
        \STATE \textit{Q update.} For each agent $i$, minimize
        \[
            \calL_i^Q(\psi_i^{k-1}) \gets \sum_{(\bms,\bma_i,r,\bms')\in\calB_i} \Big\| \wh{Q}_i^{\psi_i^{k-1}}(\bms,\bma_i) - r - \frac{\gamma}{N}\sum_{j=1}^N \wh{Q}_i^{\bar{\psi}_i^{k-1}}(\bms',\bma_{\bms',i}^{(j)}) \Big\|^2,
        \]
        and update $\psi_i^k \gets \text{Adam}(\psi_i^{k-1}, \nabla \calL_i^Q)$, $\bar{\psi}_i^k \gets \xi\psi_i^k + (1-\xi)\bar{\psi}_i^{k-1}$.
        \STATE \textit{Policy update.} For each agent $i$, minimize the importance sampling score matching loss
        \begingroup\makeatletter\def\f@size{9}\check@mathfonts\makeatother
        \begin{equation*}\begin{split}
            {}& \calL_{\bms',i}^{\is}(\theta_i^{k-1}) \gets \frac{1}{N}\sum_{j=1}^N \exp\b{\frac{\eta}{1-\gamma}\wh{Q}_i^{\psi_i^{k}}\b{\bms',\bma_{\bms',i}^{(j)}}}  l^{\theta_i^{k-1}}_{\bms',i}(\bma_{\bms',i}^{(j)}),\\
            {}& \text{where } l_{\bms,i}^{\theta}(\bma)\coloneqq \sum_{\tau=1}^H \frac{\beta_{\tau}^2}{\sigma_{\tau}^2\alpha_{\tau}}\bbE_{\epsilon_{\tau}\sim \calN(0,I)}\norm{s_{i,\tau}^{\theta}\b{\left.\sqrt{\bar{\alpha}_{\tau}}\bma+\sqrt{1-\bar{\alpha}_{\tau}}\epsilon_{\tau}\right|\bms} - \frac{\epsilon_{\tau}}{\sqrt{1 - \bar{\alpha}_\tau}}}^2.
        \end{split}\end{equation*}
        \endgroup
        and update $\theta_i^k \gets \text{Adam}(\theta_i^{k-1}, \sum_{\bms'}\nabla\calL_{\bms',i}^{\is})$.
        \ENDFOR
    \end{algorithmic}
\end{algorithm}

\section{Simulations}\label{sec:sim}
\textbf{Experiment setup.} We implement DDPL (\Cref{alg:extended}) in JAX and evaluate on 8 continuous action tasks including 4 MaMuJoCo, 1 MPE, 1 IsaacLab, and 2 SMAX tasks. Hyperparameters are listed in \Cref{tab:hyper}, among which diffusion policy update learning rate $\eta$ is chosen after light tuning and all other hyperparameters follow those in the codebase \citep{diffpolicy_4_heavy} without tuning. We train DDPL for $5\times 10^6$, $1\times 10^7$, and $5\times 10^7$ environment steps on MPE, MaMuJoCo/IsaacLab, and SMAX tasks, respectively. Training was performed on RTX 5090 GPUs, and the training time ranges from 1 to 8 hours per seed. For each environment, we computed the average episode returns over 50 trajectories. Its mean and std. are then computed across 3 random seeds. We report the mean and std. of average episode returns during training in \Cref{fig:all}, where curves show means and shaded areas show std., and report the highest mean during training in \Cref{tab:best_performance}. Our training results are benchmarked against two on-policy baselines, HAPPO \cite{marl_1} and MAPPO \cite{marl_2}, and one off-policy baseline, HASAC \cite{marl_1}. 

\textbf{Main results.} From \Cref{tab:best_performance}, \Cref{alg:extended} attains the best mean return on 7 of 8 tasks, with a performance improvement of $25.78\%$ on IsaacLab-ShadowHandOver, $132.56\%$ on 3s\_vs\_5z, and $155.84\%$ on 6h\_vs\_8z over the best baseline. 
Moreover, the shapes of the learning curves suggest that DDPL's gains come from sufficient exploration rather than early commitment to suboptimal solutions. For example, on 3s\_vs\_5z, \Cref{alg:extended} obtains lower returns than the baselines in the initial exploration stage. After exploration, the performance then improves rapidly and converges near the maximum return of $2$. We attribute the early lag to broader exploration in the action space, which is the cost paid for discovering higher-reward equilibria that unimodal baselines never reach. Similar dynamics appear on the IsaacLab and MaMuJoCo tasks.

\begin{figure}[!t]\begin{center}
       \includegraphics[width=\linewidth]{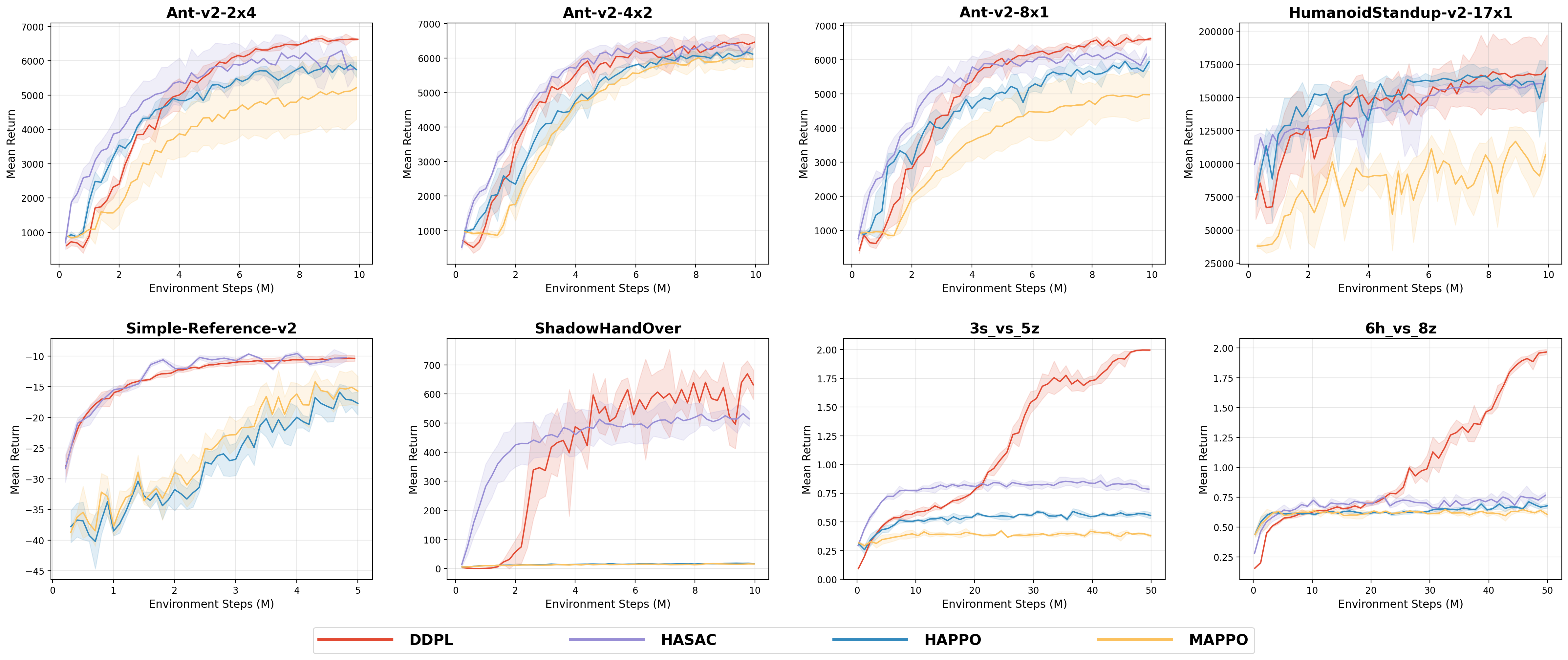}
       \caption{Average episode return during training for \Cref{alg:extended} and baselines. Curves and shaded areas show the mean and std. over three random seeds, respectively.}\label{fig:all}
\end{center}\end{figure}

\begin{table*}[!t]
  \centering
  \resizebox{\textwidth}{!}{%
  \begin{tabular}{l|c|c|c|c}
  \toprule
    \textbf{Environment} & \textbf{DDPL} & \textbf{HASAC} & \textbf{HAPPO} & \textbf{MAPPO} \\
  \midrule
  Ant-v2-2x4 & $\bm{6650.32 \pm 94.18}$ & $6301.93 \pm 422.64$ & $5881.92 \pm 87.50$ & $5215.60 \pm 922.76$ \\
  Ant-v2-4x2 & $\bm{6460.28 \pm 242.05}$ & $\bm{6408.54 \pm 161.96}$ & $6176.43 \pm 22.57$ & $6011.51 \pm 195.10$ \\
  Ant-v2-8x1 & $\bm{6633.27 \pm 121.71}$ & $6204.69 \pm 153.57$ & $5949.69 \pm 45.61$ & $4977.60 \pm 706.16$ \\
  HumanoidStandup-v2-17x1 & $\bm{172373.97 \pm 24918.25}$ & $160698.78 \pm 2993.16$ & $167491.90 \pm 10143.52$ & $116959.76 \pm 11131.91$ \\
  Simple-Reference-v2 & $-10.35 \pm 0.52$ & $\bm{-9.59 \pm 0.35}$ & $-15.89 \pm 1.31$ & $-14.20 \pm 1.92$ \\
  ShadowHandOver & $\bm{669.25 \pm 44.82}$ & $532.07 \pm 26.87$ & $18.07 \pm 1.65$ & $16.68 \pm 0.60$ \\
  3s\_vs\_5z & $\bm{2.00 \pm 0.00}$ & $0.86 \pm 0.06$ & $0.59 \pm 0.03$ & $0.42 \pm 0.01$ \\
  6h\_vs\_8z & $\bm{1.97 \pm 0.02}$ & $0.77 \pm 0.03$ & $0.71 \pm 0.03$ & $0.65 \pm 0.02$ \\
  \bottomrule
  \end{tabular}%
  }
\caption{Highest mean of the episode return with standard deviation over 3 random seeds.}
\label{tab:best_performance}
\end{table*}

\textbf{Ablation Study: diffusion vs.\ Gaussian policies.} To provide further evidence on the mechanism of DDPL's performance improvement, we probe the diffusion policies from DDPL during training and compare them with Gaussian policies from DDPL with the Gaussian policy class. \Cref{fig:ablate} (left) shows the average return, and \Cref{fig:ablate} (right) shows the action distribution at a representative state for two action dimensions from two agents. The Gaussian variant converges early to a suboptimal equilibrium with episode return $\approx 0.65$, and subsequent updates rarely changes its unimodal distribution. Contrarily, DDPL with diffusion policies maintains multi-modal action distributions for exploration in early training, discovers higher-reward modes the Gaussian variant never visits, and converges to a qualitatively different equilibrium with return $\approx 2$. This result matches our theoretical analysis in \Cref{sec:alg}, verifying that the multimodality of diffusion policies are essential to more effective exploration and improved performance.

\textbf{Sensitivity analysis.} We test DDPL's sensitivity to two diffusion hyperparameters on 3s\_vs\_5z: the number of denoising steps $H$ and the noise schedule. \Cref{fig:sensitivity} shows that final return is stable across $H \in \{10, 20, 30\}$ (with cosine schedule) and across linear, cosine, and VP schedules (with $H = 20$). This indicates that DDPL's performance does not rely on careful tuning of these hyperparameters.

\textbf{Training and inference time.} 
We measure the wall-clock time of \Cref{alg:extended}, averaged 
over 10000 training and inference steps on IsaacLab Bi-ShadowHandOver with an AMD EPYC 9275F CPU (24 cores, 48 threads) and 
an NVIDIA RTX 5090 GPU (32GB memory). We compare with DACER \citep{diffpolicy_4_heavy} and QVPO \citep{qvpo}, two single-agent diffusion-based algorithms applied to the environment in a centralized manner, as well as DDPL with the diffusion policy class replaced by Gaussian policy class. All algorithms are implemented in JAX with inference and training jitted, and all diffusion 
policies use 20 denoising steps. Results are shown in \Cref{tab:timing}. 
The inference time of DDPL, which remains below 1 ms, is moderate. The training overhead is 39\% smaller than DACER and 81\% faster than QVPO, which we attribute to ISSM's efficient online training.

\begin{figure*}[!t]\begin{center}
       \includegraphics[width=\linewidth]{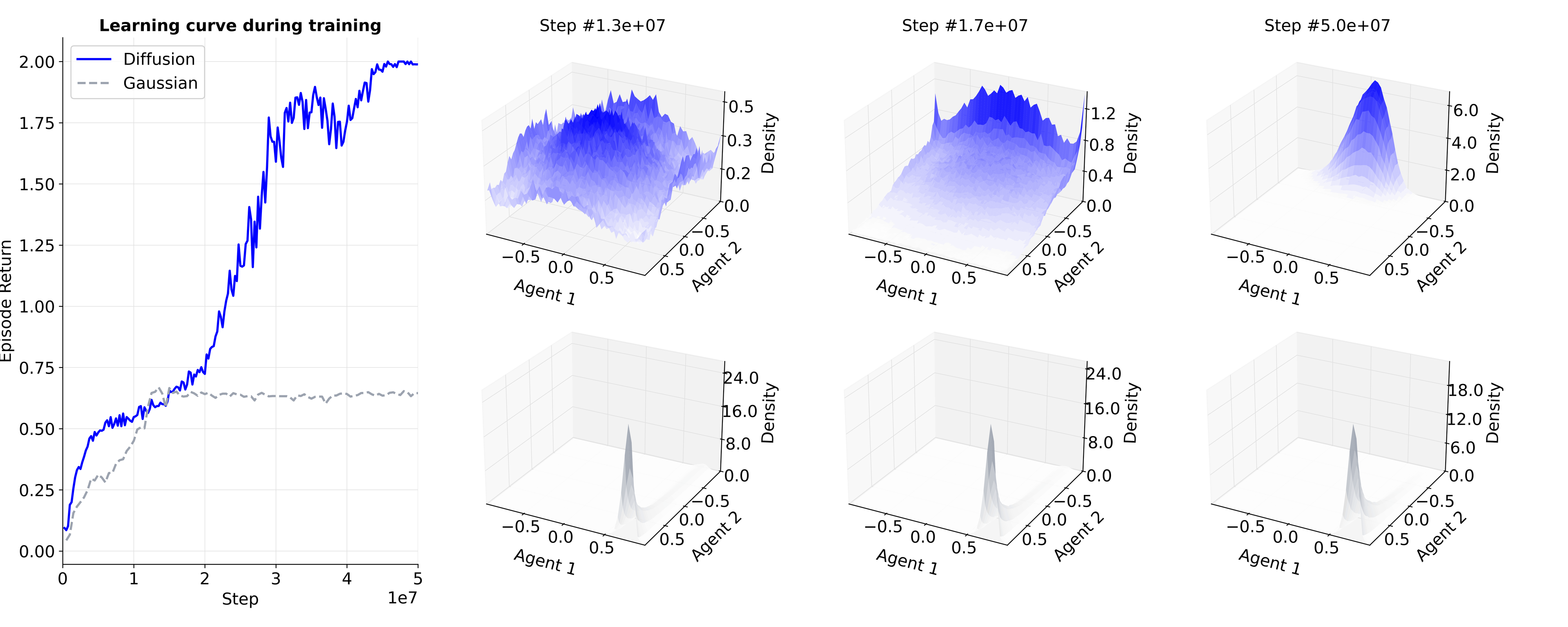}
       \caption{LEFT: the average episode return of DDPL with diffusion and Gaussian policies; RIGHT: the action distribution at a representative state for two action dimensions from two agents.}
       \label{fig:ablate}
\end{center}\end{figure*}

\newsavebox{\sensfigbox}
\begin{lrbox}{\sensfigbox}%
\begin{minipage}{0.55\linewidth}%
    \centering
    \includegraphics[width=0.85\linewidth]{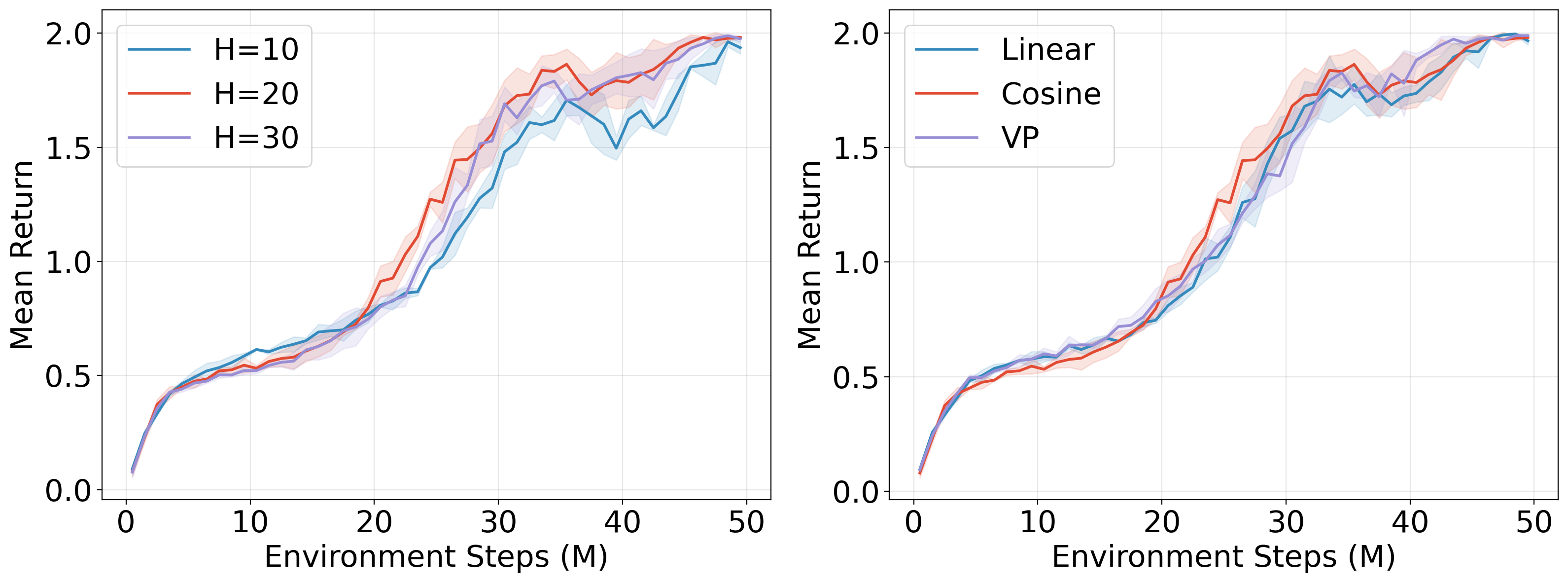}\\[-0.2em]
    \captionof{figure}{Sensitivity to denoising steps $H$ (left) and noise schedule (right) in SMAX-3s\_vs\_5z.}
    \label{fig:sensitivity}
\end{minipage}%
\end{lrbox}

\begin{figure}[h]
\centering
\usebox{\sensfigbox}\hfill
\begin{minipage}[c][\dimexpr\ht\sensfigbox+\dp\sensfigbox\relax][c]{0.4\linewidth}
    \centering
    \setlength{\tabcolsep}{4pt}\small
    \begin{tabular}{lcc}
    \toprule
    Algorithm & Train (ms) & Inf. (ms) \\
    \midrule
    DDPL & 0.820 & 3.684 \\
    DACER       & 0.814 & 6.045 \\
    QVPO        & 0.777 & 19.720 \\
    DDPL-Gaussian    & 0.374 & 2.599 \\
    \bottomrule
    \end{tabular}\\[10pt]
    \captionof{table}{Training and inference time in IsaacLab Bi-ShadowHand.}
    \label{tab:timing}
\end{minipage}
\end{figure}

\section{Conclusions \& Limitations}
This paper studies exploration in cooperative MARL through the lens of policy expressiveness. We show that a common implementation choice of projecting energy-based policies into Gaussian policy class can substantially hinder exploration and lead to degraded performance. The failure mode becomes more severe as the number of agents grows. To address this, we propose decentralized diffusion policy learning (DDPL), which parameterizes decentralized policies with diffusion models and trains them online via a novel importance sampling score matching (ISSM) objective with provable sample complexity and NE-gap guarantees.
Empirically, DDPL improves performance and sample efficiency on representative continuous-action MARL benchmarks (MPE, MaMuJoCo, IsaacLab, and SMAX), highlighting the benefit of multi-modal diffusion policies for enhanced exploration and discovering higher-reward equilibria.

\textbf{Limitations.} DDPL parameterizes fully decentralized policies. Although they require minimal communication, they do not exploit network MARL settings where low-cost communication with neighboring agents is encouraged. Extending the framework to such partially decentralized policies is a natural next step. Also, the paper focuses on the identical-interest setting. Extending to general-sum or competitive settings is left to future work.


\bibliographystyle{plainnat}
\bibliography{ref}

\appendix
\section{Theoretical Results for \Cref{subsc:example}}
\begin{proof}[Proof of \Cref{lem:sa_failure}]
    For simplicity, we omit all subscripts $i$ in this proof.

    \textbf{Step 1.} We first calculate the KL divergence between any Gaussian policy $\pi(\cdot|\bms_0)=\calN(\mu,\sigma^2)$ and $\tilde{\pi}^{k+1}(\cdot|\bms_0) = \pi^k(\cdot|\bms_0)\exp\b{\frac{\eta}{1-\gamma}\htQ^{\pi^k}(\bms_0,\cdot)} / Z$, where $Z$ is the normalizing factor.
    \begin{equation}\begin{split}
        \calL(\mu,\sigma^2) \coloneqq {}& \kl\b{\pi(\cdot|\bms_0)\|\tilde{\pi}^{k+1}(\cdot|\bms_0)} = \bbE_{\bma\sim\pi(\cdot|\bms_0)}\b{\log\frac{\pi(\bma|\bms_0)}{\tilde{\pi}^{k+1}(\bma|\bms_0)}}\\
        = {}& \bbE_{\bma\sim\pi(\cdot|\bms_0)}\bs{\log\frac{\pi(\bma|\bms_0)}{\pi^k(\bma|\bms_0)}} - \frac{\eta}{1-\gamma}\bbE_{\bma\sim\pi(\cdot|\bms_0)}\bs{\htQ^{\pi^k}(\bms_0,\bma)} + \log Z\\
        = {}& \log\b{\frac{\sigma^k}{\sigma}} + \frac{\sigma^2+\b{\mu-\mu^k}^2}{2(\sigma^k)^2} - \frac{1}{2} - \frac{\eta}{1-\gamma}\bbE_{\bma\sim\pi(\cdot|\bms_0)}\bs{\htQ^{\pi^k}(\bms_0,\bma)} + \log Z.
    \end{split}\end{equation}
    Since the environment does not reach $\bms_1$ in epochs $[k]$, and $\htQ_1^{\pi^k_1}$ is accurate for $\calM_{\text{sa}}$ with reward $r^\neg_{\text{sa}}(\bms, \bma) = \mathbb{I}\b{\bms = \bms_0}r_{\text{start}}(\bma)$, we have $\htQ^{\pi^k}(\bms_0, \bma)=r_{\text{start}}(\bma)$. We can then derive the following closed-form solution for the second last term:
    \begin{equation*}\begin{split}
        {}& \bbE_{\bma\sim\pi(\cdot|\bms_0)}\bs{\htQ^{\pi^k}(\bms_0,\bma)}\\
        = {}& \bbE_{\bma\sim\pi(\cdot|\bms_0)}\bs{-\nu\bma^2 + c_0\exp\b{-\alpha(\bma+\beta)^2} + c_0\exp\b{-\alpha(\bma-\beta)^2}}\\
        = {}& -\nu\b{\mu^2 + \sigma^2} + \frac{c_0}{\sqrt{1+2\alpha\sigma^2}}\exp\b{-\frac{\alpha(\mu-\beta)^2}{1+2\alpha\sigma^2}} + \frac{c_0}{\sqrt{1+2\alpha\sigma^2}}\exp\b{-\frac{\alpha(\mu+\beta)^2}{1+2\alpha\sigma^2}}.
    \end{split}\end{equation*}

    \textbf{Step 2.} The updated mean and covariance pair $(\mu^{k+1},(\sigma^{k+1})^2)$ is the first order stationary point of loss function $\calL(\mu,\sigma^2)$. To analyze $\mu^{k+1}$, we first derive the gradient with respect to $\mu$ and evaluate it at $\mu=\mu^k=0$. 
    \begin{equation}\begin{split}
        {}& \frac{\partial}{\partial\mu}\calL(\mu,\sigma^2) = \frac{\mu-\mu^k}{(\sigma^k)^2} - \frac{\eta}{1-\gamma}\frac{\partial}{\partial \mu}\bbE_{\bma\sim\pi(\cdot|\bms_0)}\bs{\htQ^{\pi^k}(\bms_0,\bma)}.
    \end{split}\end{equation}
    For the last term, we have:
    \begin{equation*}\begin{split}
        {}& \left.\frac{\partial}{\partial \mu}\bbE_{\bma\sim\pi(\cdot|\bms_0)}\bs{\htQ^{\pi^k}(\bms_0,\bma)}\right|_{\mu=\mu^k}\\
        = {}& -2\nu\mu + \frac{2c_0}{\sqrt{1+2\alpha\sigma^2}}\frac{\alpha\b{\beta-\mu}}{1+2\alpha\sigma^2}\exp\b{-\frac{\alpha(\mu-\beta)^2}{1+2\alpha\sigma^2}} \left. + \frac{2c_0}{\sqrt{1+2\alpha\sigma^2}}\frac{\alpha\b{-\beta-\mu}}{1+2\alpha\sigma^2}\exp\b{-\frac{\alpha(\mu+\beta)^2}{1+2\alpha\sigma^2}}\right|_{\mu=\mu^k}\\
        = {}& 0.
    \end{split}\end{equation*}
    Substituting back, we have:
    \begin{equation}\begin{split}
        {}& \left.\frac{\partial}{\partial\mu}\calL(\mu,\sigma^2)\right|_{\mu=\mu^k} = 0.
    \end{split}\end{equation}
    Therefore, starting from $\mu^k=0$, we have $\mu^{k+1}=0$.

    \textbf{Step 3.} Next, we derive $\partial\calL(\mu,\sigma^2)/\partial\sigma^2$:
    \begin{equation}\begin{split}\label{eq:opt_sigma}
        \frac{\partial}{\partial \sigma^2}\calL(\mu,\sigma^2) = {}& -\frac{1}{2\sigma^2} + \frac{1}{2(\sigma^k)^2} - \frac{\eta}{1-\gamma}\frac{\partial}{\partial \sigma^2}\bbE_{\bma\sim\pi(\cdot|\bms_0)}\bs{\htQ^{\pi^k}(\bms_0,\bma)}.
    \end{split}\end{equation}

    For the last term, we have:
    \begin{equation}\begin{split}
        \frac{\partial}{\partial \sigma^2}\bbE_{\bma\sim\pi(\cdot|\bms_0)}\bs{\htQ^{\pi^k}(\bms_0,\bma)}= {}& -\nu + c_0\frac{\alpha\b{2\alpha(\mu-\beta)^2-1-2\alpha\sigma^2}}{\b{1+2\alpha\sigma^2}{^{5/2}}} \exp\b{-\frac{\alpha(\mu-\beta)^2}{1+2\alpha\sigma^2}}\\
        {}& + c_0\frac{\alpha\b{2\alpha(\mu+\beta)^2-1-2\alpha\sigma^2}}{\b{1+2\alpha\sigma^2}{^{5/2}}} \exp\b{-\frac{\alpha(\mu+\beta)^2}{1+2\alpha\sigma^2}}.
    \end{split}\end{equation}
    Since $\mu$ stays at $0$, the above simplifies to 
    \begin{equation}\begin{split}
        {}& \frac{\partial}{\partial \sigma^2}\bbE_{\bma\sim\pi(\cdot|\bms_0)}\bs{\htQ^{\pi^k}(\bms_0,\bma)}= -\nu\\
        {}& + 2c_0\frac{\alpha\b{2\alpha\beta^2-1-2\alpha\sigma^2}}{\b{1+2\alpha\sigma^2}{^{5/2}}} \exp\b{-\frac{\alpha\beta^2}{1+2\alpha\sigma^2}}\\
    \end{split}\end{equation}
    Substituting back into \Cref{eq:opt_sigma} and setting the gradient to $0$, we have:
    \begin{equation}\begin{split}
        \frac{1}{2\sigma^2} = {}& \frac{1}{2(\sigma^k)^2}
        - \frac{\eta}{1-\gamma}\frac{\partial}{\partial \sigma^2}\bbE_{\bma\sim\pi(\cdot|\bms_0)}\bs{\htQ^{\pi^k}(\bms_0,\bma)}\\
        = {}& \frac{1}{2(\sigma^k)^2}
        + \frac{\eta}{1-\gamma}\b{\nu - 2c_0\frac{\alpha\b{2\alpha\beta^2-1-2\alpha\sigma^2}}{\b{1+2\alpha\sigma^2}{^{5/2}}} \exp\b{-\frac{\alpha\beta^2}{1+2\alpha\sigma^2}}}.
    \end{split}\end{equation}
    With $\eta=0.01(1-\gamma), \gamma=0.1, \nu=10, c_0=2.7, \alpha=250$, and $ \beta=0.5$, we have
    \begin{equation}\begin{split}
        \frac{1}{\sigma^2} = {}& \frac{1}{(\sigma^k)^2}
        + 0.02\left(10 - 1350\frac{124-500\sigma^2}{\b{1+500\sigma^2}^{5/2}}\cdot\exp\b{-\frac{62.5}{1+500\sigma^2}}\right)\\
        \geq {}& \frac{1}{(\sigma^k)^2} + 0.02\b{10-3.6}\geq \frac{1}{(\sigma^k)^2} + 0.128.
    \end{split}\end{equation}
    Since the above inequality holds for all stationary points $\sigma$, we know that
    \begin{equation}\begin{split}\label{eq:sigma}
        \frac{1}{\b{\sigma^{k}}^2} \geq \frac{1}{\b{\sigma^{k-1}}^2} + 0.128 \geq \cdots \geq \frac{1}{\b{\sigma^0}^2} + 0.128k.
    \end{split}\end{equation}

    \textbf{Step 4.} Finally, we lower bound the probability that the algorithm never reaches $\bms_1$. For any epoch $k$ before $\bms_1$ is reached for the first time, consider policy $\pi^k = \calN(0,(\sigma^k)^2)$ with $(\sigma^k)^2$ satisfying \Cref{eq:sigma}. We upper bound the probability of reaching state $\bms_1$ with a single action sample from $\pi^k$, denoted by $p_k$. 
    \begin{equation}\begin{split}
        p_k = {}& \int_{-\beta-\frac{1}{4}}^{-\beta+\frac{1}{4}} \frac{1}{\sqrt{2\pi (\sigma^k)^2}} \exp\b{-\frac{\bma^2}{2\b{\sigma^k}^2}}\cdot c_p\exp\b{-\alpha(\bma+\beta)^2}\d\bma\\
        \overset{(i)}{=} {}& \frac{C'}{\sqrt{2\pi (\sigma')^2}} \int_{-\beta-\frac{1}{4}}^{-\beta+\frac{1}{4}} \exp\b{-\frac{(\bma-\mu')^2}{2\b{\sigma'}^2}}\d \bma\\
        \leq {}& C'\Phi\b{\frac{-\beta+\frac{1}{4}-\mu'}{\sigma'}}\\
    \end{split}\end{equation}
    In $(i)$ we have defined
    \begin{equation}\begin{split}
        \mu' = {}& -\frac{2\alpha\beta(\sigma^k)^2}{1+2\alpha(\sigma^k)^2},\quad (\sigma')^2 = \frac{\b{\sigma^k}^2}{1+2\alpha(\sigma^k)^2}, \quad C' = \frac{c_p}{\sqrt{1+2\alpha(\sigma^k)^2}}\exp\b{-\frac{\alpha\beta^2}{1+2\alpha(\sigma^k)^2}}.
    \end{split}\end{equation}
    For $k \geq 150000$, we have that $(\sigma^{k})^2 \leq 1/(10+0.128k) \leq 1/19210$ and that $-\mu' \leq 0.05$. This leads to
    \begin{equation}\begin{split}
        p_k \leq {}& C'\Phi\b{-\frac{1}{5\sigma'}} \leq C'\Phi\b{-\frac{1}{5}\sqrt{\frac{1+2\alpha\b{\sigma^k}^2}{\b{\sigma^k}^2}}} \leq C'\Phi\b{-\frac{1}{5}\sqrt{0.128k+510}}\\
        \leq {}& \frac{5C'\exp\b{-\b{0.00256k+10.2}}}{\sqrt{2\pi\b{0.128k+510}}}.
    \end{split}\end{equation}
    In the last line, we have used $\Phi(-x) \leq \exp(-x^2/2)/(\sqrt{2\pi}x)$. This is because $\Phi(-x) = \frac{1}{\sqrt{2\pi}} \int_x^\infty \exp\b{-\frac{t^2}{2}}\d t \leq \frac{1}{\sqrt{2\pi}} \int_x^\infty \frac{t}{x}\exp\b{-\frac{t^2}{2}}\d t = \frac{1}{\sqrt{2\pi}x}\exp\b{-t^2/2}|^x_\infty = \frac{1}{\sqrt{2\pi}x}\exp\b{-x^2/2}$ for $x > 0$.

    We now use the above inequality to lower bound the probability of not reaching state $\bms_1$ in all epochs, denoted by $\overline{p}_\infty$.
    \begin{equation}\begin{split}
        \overline{p}_{\infty} \geq {}& \prod_{k=0}^\infty \b{1 - p_k}^{32} = \prod_{k=0}^\infty \exp\b{\ln\b{1 - p_k}}^{32} \overset{(i)}{\geq}  \exp\b{-32\sum_{k=0}^\infty\frac{p_k}{1-p_k}}\\
        = {}& \exp\b{-32\sum_{k=0}^{149999}\frac{p_k}{1-p_k}}\exp\b{-32\sum_{k=150000}^{\infty}\frac{p_k}{1-p_k}}\\
        \geq {}& \exp\b{-32\sum_{k=0}^{149999}\frac{p_k}{1-p_k}}\exp\b{-\frac{32}{0.9999}\sum_{k=150000}^{\infty}p_k}\\
    \end{split}\end{equation}
    In $(i)$, we have used $\ln(1-x) \geq -x/(1-x)$ for $x\in(0,1)$.
    We bound $\sum_{k=150000}^{\infty}p_k$ as follows.
    \begin{equation}\begin{split}
        \sum_{k=150000}^\infty p_k \leq {}& \int_{149999}^\infty \frac{5C'\exp\b{-\b{0.00256x+10.2}}}{\sqrt{2\pi\b{0.128x+510}}}\d x\\
        \leq {}& \frac{5C'\exp\b{-10.2}}{\sqrt{1020\pi}} \int_{149999}^\infty \exp\b{-0.00256x}\d x
        \leq 10^{-100}.
    \end{split}\end{equation}
    Therefore, 
    \begin{equation}\begin{split}
        \overline{p}_{\infty} \geq {}& \exp\b{-32\sum_{k=0}^{149999}\frac{p_k}{1-p_k}}\exp\b{-\frac{32}{0.9999}\sum_{k=150000}^{\infty}p_k}
        \geq 0.55.
    \end{split}\end{equation}
\end{proof}

\begin{proof}[Proof of Proposition \ref{lem:ma_failure}]
    We let $E_i$ be the event that $[\bms]_i=1$ at least once when executing \Cref{alg:exact}, and let $k_i$ denote the first epoch when $[\bms]_i=1$. Then
    \begin{equation}\begin{split}
        {}& \bbP\b{\text{visit }\bms=\bm{1}_n \text{ at least once}} \leq \bbP\b{\bigcap_{i\in[n]}E_i}.
    \end{split}\end{equation}
    Notice that event $E_i^c$, i.e., the agents never visit $[\bms]_i=1$, only depends on the individual policy $\pi_i^k(\cdot|\bms_{\text{start}})$ for epoch $k\leq k_i$. Moreover, the individual policy $\pi_i^k$ with $k\leq k_i$ can be expressed recursively as
    \begin{equation}\begin{split}
        {}& \pi_i^k(\cdot|\bms_{\text{start}})= \arg\min_{\pi_i}\text{KL}\b{\pi_i\|\tilde{\pi}_i^k},\quad \tilde{\pi}_i^k(\cdot)\propto \pi_i^{k-1}(\cdot|\bms_0)\exp\b{\frac{\eta \htQ_i^{\pi^{k-1}}(\bms_0,\cdot)}{(1-\gamma)}}.
    \end{split}\end{equation}
    Since $\htQ_i^{\pi^{k-1}}$ is accurate for $\calM_{\text{ma}}$ reward function $r_{\text{ma}}^\neg(\bms, \bma) = \mathbb{I}\b{\bms = -\bm{1}_n}\sum_{j=1}^n r_{\text{start}}(\bma_j) + \mathbb{I}\b{\bms\neq -\bm{1}_n}\,c_1\sum_{j\neq i}[\bms]_j$, policy $\pi_i$ only contributes through term $r_{\text{start}}(\bma_j)$. Therefore, $\htQ_i^{\pi^{k-1}}(\bms_0,\cdot) = r_{\text{start}}(\bma_i) + g(\pi_{-i})$ for some function $g$. Substituting back gives:
    \begin{equation}\begin{split}
        {}& \tilde{\pi}_i^k(\cdot)\propto \pi_i^{k-1}(\cdot|\bms_0)\exp\b{\frac{\eta \htQ_i^{\pi^{k-1}}(\bms_0,\cdot)}{(1-\gamma)}} \propto \pi_i^{k-1}(\cdot|\bms_0)\exp\b{\frac{\eta r_{\text{start}}(\bma_i)}{(1-\gamma)}}.
    \end{split}\end{equation}
    Therefore, the policy update of agent $i$ is independent of other agents for $k\leq k_i$. Thus, events $\{E_i^c\}_{i\in[n]}$ remain independent of each other, and $\bbP(E_i^c) \geq 0.55$ by \Cref{lem:sa_failure}.
    
    As a result, events $\{E_i\}_{i\in[n]}$ are also independent and $\bbP(E_i) \leq 0.45$ for all $i\in[n]$. This gives    
    \begin{equation}\begin{split}
        {}& \bbP\b{\text{visit }\bms=\bm{1}_n \text{ at least once}} \leq \bbP\b{\bigcap_{i\in[n]}E_i} = \prod_{i\in[n]} \bbP\b{E_i} \leq 0.45^n.
    \end{split}\end{equation}

\end{proof}

\section{Theoretical Results for \Cref{sec:is}}
\subsection{Importance Sampling Score Matching Loss}\label{sec:issmloss}
To show the equivalence, we first rewrite the importance sampling loss as follows
\begin{equation}\begin{split}
    {}& \calL_{i,\bms}^{\is}(\theta_i) \coloneqq \frac{1}{N}\sum_{j=1}^N \frac{\pi_i^{\star,k}(\bma_i^{(j)}|\bms)}{\tilp(\bma_i^{(j)}|\bms)} l^{\theta_i}_{i}(\bma_i^{(j)}|\bms),\\
    {}& \text{where}\quad l^{\theta_i}_{i}(\bma_i^{(j)}|\bms) \coloneqq \sum_{\tau=1}^H \frac{\beta_{\tau}^2}{\sigma_{\tau}^2\alpha_{\tau}}\bbE_{\epsilon_{\tau}\sim \calN(0,I)}\norm{s_{i,\tau}^{\theta_i}\b{\left.\sqrt{\overline{\alpha}_{\tau}}\bma_i^{(j)}+\sqrt{1-\overline{\alpha}_{\tau}}\epsilon_{\tau}\right|\bms} - \frac{\epsilon_{\tau}}{\sqrt{1 - \overline{\alpha}_\tau}}}^2.
\end{split}\end{equation}
In expectation, $\calL_{i,\bms}^{\is}(\theta_i)$ is equivalent to the original denoising score matching loss $\calL_{i,\bms}(\theta_i)$ because
\begin{equation}\begin{split}
    \bbE_{\bma_i^{(j)}\sim \tilp(\cdot|\bms)}\bs{\calL_{i,\bms}^{\is}(\theta_i)}
    = {}& \frac{1}{N}\sum_{j=1}^N\int\tilp(\bma_i^{(j)}|\bms)\frac{\pi_i^{\star,k}(\bma_i^{(j)}|\bms)}{\tilp(\bma_i^{(j)}|\bms)} \cdot l^{\theta_i}_{i}(\bma_i^{(j)}|\bms) \d \bma_i^{(j)}\\
    = {}& \int\pi_i^{\star,k}(\bma_i|\bms)l^{\theta_i}_{i}(\bma_i|\bms) \d \bma_i = \bbE_{\bma_i\sim\pi_i^{\star,k}(\cdot|\bms)}\bs{l^{\theta_i}_{i}(\bma_i|\bms)}\\
    = {}& \calL_{i,\bms}(\theta_i).
\end{split}\end{equation}

\subsection{Proof of \Cref{thm:isflow}}\label{sec:is_proof_inf_inf}
The variance schedule considered in this theorem is commonly studied in diffusion model literature \citep{proof_diff}, and is defined for sufficiently large constants $c',c''$:
\begin{equation}\begin{split}\label{eq:scheduler}
    {}& \beta_1= \frac{1}{H^{c'}}, \quad \beta_{\tau+1} = \frac{c''\log H}{H} \min\left\{\beta_1\b{1+\frac{c''\log H}{H}}^{\tau}, 1\right\},\\
    {}& \alpha_\tau = 1 - \beta_\tau, \quad \bar{\alpha}_\tau = \prod_{h=1}^\tau \alpha_h, \quad \sigma_{\tau}^2=\frac{\beta_{\tau}}{\alpha_{\tau}}.
\end{split}\end{equation}

To proceed, we first introduce the following assumptions on the score function class $\calS_{\Theta}\coloneqq \left\{(\tau,\bma)\mapsto s^{\theta}_{\tau}(\bma): \tau\in[H], \bma\in\calA_i, \theta\in\Theta\right\}$.

\begin{assumption}\label{assmp:score_comp}
(i). \emph{Expressiveness.} There exists $\theta^\star \in \Theta$ such that $s^{\theta^\star}_{\tau}(\bma_{\tau}) = \nabla_{\bma_{\tau}}p_{\tau}(\bma_{\tau})$.

(ii). \emph{Complexity.} The $\kappa$-covering number of $\calS_{\Theta}$ \footnote{The $\kappa$-covering number of function class $\calS_{\Theta}$ is the cardinality of the smallest parameter subset $\tilde{\Theta} \subset \Theta$ such that for every $\theta \in \Theta$, there exists $\tilde{\theta} \in \tilde{\Theta}$ such that $\|s^{\theta}_{\tau}(\bma) - s^{\tilde{\theta}}_{\tau}(\bma)\| \leq \kappa$ for all $\tau \in [H]$ and $\bma$ in the support.
}, denoted by $\calC(\calS_{\Theta},\kappa)$, satisfies $\log\calC(\calS_{\Theta},\kappa) \leq c_1 \log(c_1/\kappa)$ for constant $c_1 \geq 1$ and all $\kappa > 0$.

(iii). \emph{Linear growth.} $\norm{s_{\tau}^{\theta}(\bma)} \leq c_2(1 + \norm{\bma})$ for constant $c_2\geq 1$, all $\theta \in \Theta$, $\tau \in [H]$, and $\bma \in \bbR^d$.
\end{assumption}

\Cref{assmp:score_comp} is mild and standard in function approximation literature. Part (i) on expressiveness is satisfied by common neural network architectures, which are universal function approximators. Parts (ii) 
and (iii) are satisfied by various popular neural network architectures including feedforward networks, convolutional networks, and transformers, where (ii) is established in related literature (\citep{gen_cover, transformer_cover}) and (iii) follows by their Lipschitzness (\citep{gen_lipschitz, 
attention_lipschitz_2}). Moreover, the true score function $s^{\theta^\star}_{\tau}(\bma_{\tau})$ is often assumed to be Lipschitz in the rich diffusion literature \citep{score_lipschitz,score_lipschitz_2}, implying the linear growth of the true score function. This further justifies the validity of part (iii).

We also note that Part (i) is mainly for simplicity of the proof. We can easily extend it to $\epsilon$-realizability, which assumes the existence of $\theta^\star\in\Theta$ such that $s^{\theta^\star}_{\tau}(\bma_{\tau})$ is only $\epsilon$-accurate. This will lead to an additional error term of size $\epsilon$ in the final squared TV distance bound.

\begin{proof}
    Recall that the target distribution is $\pi_i^{\star,k}(\cdot|\bms)$, the true score function is $s_{i,\tau}^{\theta^\star}(\cdot|\bms)$,
    and the sampling distribution is $\tilp(\cdot|\bms)$. The learned action distribution with or without projection onto the bounded action space $\calA_i$ are denoted by $\pi^{\wh{\theta}_i}_i(\cdot|\bms)$ or $\pi_i^{\wh{\theta}_i, \noclip}(\cdot|\bms)$, respectively. For simplicity, we omit subscript $i$ and the dependence on state $\bms$ and epoch $k$, reducing the above notations to $\pi^\star(\cdot)$, $s_{\tau}^{\theta^\star}(\cdot)$, $\tilp(\cdot)$, $\pi^{\wh{\theta}}(\cdot)$, and $\pi^{\wh{\theta}, \noclip}(\cdot)$. 
    
    Moreover, we define importance weight $w(\bma)$, per-sample loss $l^{\theta}(\bma)$, and per-sample weighted loss $f^{\theta}(\bma)$ as follows. All are supported on $\calA_i\subseteq\{\bma: \norm{\bma}\leq c_b\}$.
    \begin{equation}\begin{split}\label{eq:main_1}
        f^{\theta}(\bma) \coloneqq {}& w(\bma)l^{\theta}(\bma),\quad w(\bma) \coloneqq \frac{\pi^\star(\bma)}{\tilp(\bma)},\\
        l^{\theta}(\bma) \coloneqq {}& \sum_{\tau=1}^H \frac{\beta_{\tau}^2}{\sigma_{\tau}^2\alpha_{\tau}} \bbE_{\epsilon_{\tau}\sim \calN(0,I)}\norm{s_{\tau}^{\theta}\b{\sqrt{\overline{\alpha}_{\tau}}\bma+\sqrt{1-\overline{\alpha}_{\tau}}\epsilon_{\tau}} - \frac{\epsilon_{\tau}}{\sqrt{1 - \bar{\alpha}_\tau}}}^2\\
        = {}& \sum_{\tau=1}^H \beta_{\tau} \bbE_{\epsilon_{\tau}\sim \calN(0,I)}\norm{s_{\tau}^{\theta}\b{\sqrt{\overline{\alpha}_{\tau}}\bma+\sqrt{1-\overline{\alpha}_{\tau}}\epsilon_{\tau}} - \frac{\epsilon_{\tau}}{\sqrt{1 - \bar{\alpha}_\tau}}}^2\\
    \end{split}\end{equation}
    The loss can then be defined as follows for $\{\bma^{(j)} \overset{\text{iid}}{\sim} \tilp\}_{j\in[N]}$.
    \begin{equation}\begin{split}\label{eq:L_inf}
        {}& \calL(\theta) \coloneqq \bbE_{\bma\sim\tilp}[f^{\theta}(\bma)] = \bbE_{\bma\sim \pi^\star}[l^{\theta}(\bma)], \quad \wh{\calL}(\theta) \coloneqq \frac{1}{N}\sum_{j=1}^N f^{\theta}(\bma^{(j)}).
    \end{split}\end{equation}
    We will analyze the solution to the optimization problem $\wh{\theta} \coloneqq \mathop{\arg\min}_{\theta\in\Theta} \wh{\calL}(\theta)$ and the resulting projected policy $\pi^{\wh{\theta}}(\cdot)$.
    The proof proceeds in two steps. In Step 1, 
    we establish a concentration bound on $\|\wh{\calL}(\theta) - \calL(\theta)\|$ for all $\theta\in\Theta$, and utilize the above concentration to upper bound the population loss difference $\calL(\wh{\theta})-\calL(\theta^\star)$. In Step 2, we connect this loss difference to the TV distance $D_{\tv}(\pi^\star(\cdot), \pi^{\wh{\theta}}(\cdot))$.

    \textbf{Step 1: Concentration of $\bm{\wh{\calL}(\theta)-\calL(\theta)}$.} 
    We start from a $1/\sqrt{N}$-cover of function class $\calF_{\Theta}=\{f^{\theta}(\cdot): \theta\in\Theta\}$, denoted by $\calF_{\tilde{\Theta}}$. By \Cref{lem:covering_of_loss}, we know that $|\calF_{\tilde{\Theta}}| \leq (3(c'+c'')c_1c_2c_w (1+c_b+2\sqrt{d})\sqrt{N}\log H)^{c_1}$.

    We first bound the mean and variance of $f^{\tilde{\theta}}(\bma)$ for any $\tilde{\theta}\in\tilde{\Theta}$ and $\bma\in\overline{\calA}$. 
    For its mean, we have the following bound by \Cref{lem:bounded_l}, 
    \begin{equation}\begin{split}
        \|f^{\tilde{\theta}}(\bma) - \bbE f^{\tilde{\theta}}(\bma)\|\leq {}& \max_{\bma\in\overline{\calA}}\|f^{\tilde{\theta}}(\bma)\| \leq c_w \max_{\bma\in\overline{\calA}}\|l^{\tilde{\theta}}(\bma)\|\\
        \leq {}& \underbrace{3(c'+c'')c_2^2(1+c_b+2\sqrt{d})^2\log H}_{\kappa_1}c_w.
    \end{split}\end{equation}
    Here the first inequality is because $f^{\tilde{\theta}}(\bma) > 0$. Moreover, for its variance, we have
    \begin{equation}\begin{split}
        \bbE_{\bma\sim\tilp}\b{\norm{f^{\tilde{\theta}}(\bma) - \bbE f^{\tilde{\theta}}(\bma)}^2} \leq {}& \bbE_{\bma\sim\tilp}\b{\norm{f^{\tilde{\theta}}(\bma)}^2} \overset{(i)}{\leq} \sqrt{\bbE_{\bma\sim\tilp} (w(\bma)^4)}\sqrt{\bbE_{\bma\sim\tilp} (l^{\tilde{\theta}}(\bma)^4)}\\
        \overset{(ii)}{\leq} {}& \kappa_1^2\sqrt{\bbE_{\bma\sim\tilp} (w(\bma)^4)}.
    \end{split}\end{equation}
    Here $(i)$ is by Cauchy-Schwarz and $(ii)$ is by \Cref{lem:bounded_l}.
    Therefore, by the Bernstein's inequality, we know that with probability at least $1-\delta'$ for any $\delta'\in (0,1/|\calF_{\tilde{\Theta}}|)$,
    \begin{equation}\begin{split}\label{eq:inter1}
        \norm{\frac{1}{N}\sum_{j=1}^N f^{\tilde{\theta}}(\bma^{(j)}) - \bbE f^{\tilde{\theta}}(\bma)}
        \leq {}& \underbrace{\sqrt{2\kappa_1^2\log (2/\delta')}\frac{\sqrt[4]{\bbE_{\bma\sim\tilp} (w(\bma)^4)}}{\sqrt{N}}}_{\calI_1} + \underbrace{\frac{2 \kappa_1c_w\log(2/\delta')}{3}\frac{1}{N}}_{\calI_2}.
    \end{split}\end{equation}
    Since $N \geq 2c_w^2\log(2/\delta')/9$ and $\bbE(w(\bma)^4) \geq \bbE(w(\bma))^{4} \geq 1$, we know that $\calI_1 \geq \calI_2$. Therefore,
    \begin{equation}\begin{split}\label{eq:main_3}
        \norm{\frac{1}{N}\sum_{j=1}^N f^{\tilde{\theta}}(\bma^{(j)}) - \bbE f^{\tilde{\theta}}(\bma)} \leq 2\calI_1
        = {}& 2\sqrt{2}\kappa_1\sqrt{\log \frac{2}{\delta'}}\frac{\sqrt[4]{\bbE_{\bma\sim\tilp} (w(\bma)^4)}}{\sqrt{N}}.
    \end{split}\end{equation}
    With a union bound, \Cref{eq:main_3} holds for all $\tilde{\theta}\in\tilde{\Theta}$ with probability at least $1-|\calF_{\tilde{\Theta}}|\delta' = 1-\delta$. Here we have let $\delta=|\calF_{\tilde{\Theta}}|\delta'\in(0,1)$.

    Finally, we extend the above bound on $\tilde{\theta}\in\tilde{\Theta}$ to any $\theta\in\Theta$.
    By the definition of $1/\sqrt{N}$ cover $\calF_{\tilde{\Theta}}$, for any $\theta\in\Theta$, we can find $\tilde{\theta}\in\tilde{\Theta}$ such that $\norm{f^{\theta}(\bma) - f^{\tilde{\theta}}(\bma)} \leq 1/\sqrt{N}, ~\forall\bma\in\overline{A}$. Therefore, the following holds for any $\theta\in\Theta$ with probability at least $1-\delta$,
    \begin{align*}
        \norm{\wh{\calL}(\theta)-\calL(\theta)} = {}& \norm{\frac{1}{N}\sum_{j=1}^N f^{\theta}(\bma^{(j)}) - \bbE f^{\theta}(\bma)} \leq \norm{\frac{1}{N}\sum_{j=1}^N f^{\tilde{\theta}}(\bma^{(j)}) - \bbE f^{\tilde{\theta}}(\bma)} + \frac{2}{\sqrt{N}}\\
        \leq {}& \b{2\sqrt{2}\kappa_1\sqrt{\log \frac{2|\calF_{\tilde{\Theta}}|}{\delta}} + 2}\frac{\sqrt[4]{\bbE_{\bma\sim\tilp} (w(\bma)^4)}}{\sqrt{N}}.
    \end{align*}

    By definition of $\wh{\theta}$, we know that $\wh{\calL}(\wh{\theta})-\wh{\calL}(\theta^\star) \leq 0$. Therefore, with probability at least $1-\delta$,
    \begin{equation}\begin{split}
        \calL(\wh{\theta})-\calL(\theta^\star) \leq {}& \calL(\wh{\theta}) - \wh{\calL}(\wh{\theta}) + \wh{\calL}(\wh{\theta})-\wh{\calL}(\theta^\star)+ \wh{\calL}(\theta^\star) - \calL(\theta^\star)\\
        \leq {}& \calL(\wh{\theta}) - \wh{\calL}(\wh{\theta}) + \wh{\calL}(\theta^\star) - \calL(\theta^\star) \leq 2\max_{\theta\in\Theta}\norm{\calL(\theta)-\wh{\calL}(\theta)}\\
        \leq {}& \b{4\sqrt{2}\kappa_1\sqrt{\log \frac{2|\calF_{\tilde{\Theta}}|}{\delta}} + 4}\frac{\sqrt[4]{\bbE_{\bma\sim\tilp} (w(\bma)^4)}}{\sqrt{N}}.
    \end{split}\end{equation}

    Recall that 
    \begin{equation}\begin{split}
        \calL^{\text{marginal}}(\theta) = {}& \sum_{\tau=1}^H \frac{\beta_{\tau}^2}{\sigma_{\tau}^2\alpha_{\tau}}\bbE_{\bma(\tau)}\norm{s_{\tau}^{\theta}\b{\bma(\tau)} - s_{\tau}^{\theta^\star}\b{\bma(\tau)}}^2\\
        = {}& \sum_{\tau=1}^H (1-\alpha_{\tau})\bbE_{\bma(\tau)}\norm{s_{\tau}^{\theta}\b{\bma(\tau)} - s_{\tau}^{\theta^\star}\b{\bma(\tau)}}^2.
    \end{split}\end{equation}
    By Proposition 3.1 in \cite{diffusion_1}, we know that $\calL^{\text{marginal}}(\theta)=\calL(\theta) + C_3$ for constant $C_3$ independent of $\theta$. Therefore, with probability at least $1-\delta$,
    \begin{equation}\begin{split}
        \calL^{\text{marginal}}(\wh{\theta}) \leq {}& \calL^{\text{marginal}}(\wh{\theta}) - \calL^{\text{marginal}}(\theta^\star) = \calL(\wh{\theta}) - \calL(\theta^\star)\\
        \leq {}& \b{4\sqrt{2}\kappa_1\sqrt{\log \frac{2|\calF_{\tilde{\Theta}}|}{\delta}} + 4}\frac{\sqrt[4]{\bbE_{\bma\sim\tilp} (w(\bma)^4)}}{\sqrt{N}}\\
        \leq {}& \b{4\sqrt{2}\kappa_1\sqrt{c_1\log \frac{2\kappa_1c_1c_w\sqrt{N}}{\delta}} + 4}\frac{\sqrt[4]{\bbE_{\bma\sim\tilp} (w(\bma)^4)}}{\sqrt{N}}.
    \end{split}\end{equation}
    Here the first inequality is because $\calL^{\text{marginal}}(\theta^\star) = 0$ by \Cref{assmp:score_comp} (i), and the last inequality is by $|\calF_{\tilde{\Theta}}| \leq (3(c'+c'')c_1c_2c_w (1+c_b+2\sqrt{d})\sqrt{N}\log H)^{c_1} \leq ({2\kappa_1c_1c_w\sqrt{N}}/{\delta})^{c_1}$.

    \textbf{Step 2: Connecting to $\bm{D_{\tv}(\pi^\star(\cdot), \pi^{\wh{\theta}}(\cdot))}$.}
    First, recall that $\pi^{\wh{\theta},\noclip}$ and $\pi^\star$ are policies induced by the reverse diffusion denoising process at $\tau=0$ with score functions $s^{\wh{\theta}}$ and $s^{\theta^{\star}}$, respectively. We further introduce policies $\pi^{\wh{\theta},\noclip}_1$ and $\pi^\star_1$ to represent the policies induced by the diffusion denoising process at diffusion step $\tau=1$. 
    Combining Equation (20-23 (i)) and the last inequality in \cite[Section 4]{proof_diff}, we have the following for some constant $\kappa_2\geq 1$
    \begin{equation}\begin{split}
        D_{\tv}^2\b{\pi^\star_1,\pi^{\wh{\theta},\noclip}_1} \leq {}&  \kappa_2\frac{d^2\log^6H}{H^2} + \kappa_2\sum_{\tau=2}^H \frac{1-\alpha_{\tau}}{2}\cdot \bbE_{\bma(\tau)}\bs{\norm{s_{\tau}^{\wh{\theta}}\b{\bma(\tau)}-s_{\tau}^{\theta^\star}\b{\bma(\tau)}}^2}.
    \end{split}\end{equation}

    Utilizing the one-step denoising TV bound \Cref{lem:one_step_denoising} with $\sigma^2=\beta_1/\alpha_1, \beta=\beta_1, \alpha=\alpha_1$, 
    \begin{equation}\begin{split}\label{eq:main_4}
        D_{\tv}^2\b{\pi^\star_0,\pi^{\wh{\theta},\noclip}_0} \leq {}& \frac{\beta_1}{2}\bbE_{\bma(1)}\bs{\norm{s_{1}^{\wh{\theta}}\b{\bma(1)}-s_{1}^\star\b{\bma(1)}}^2} + 2D_{\tv}^2\b{\pi^\star_1,\pi^{\wh{\theta},\noclip}_1}\\
        {}& \hspace{2em} + \frac{\beta_1c_2^2d}{\alpha_1\b{\alpha_1/\beta_1-c_2}} \\
        \overset{(i)}{\leq} {}& 3\kappa_2\frac{d^2\log^6H}{H^2} + \kappa_2\sum_{\tau=1}^H \b{1-\alpha_{\tau}}\cdot \bbE_{\bma(\tau)}\bs{\norm{s_{\tau}^{\wh{\theta}}\b{\bma(\tau)}-s_{\tau}^{\theta^\star}\b{\bma(\tau)}}^2}\\
        = {}& 3\kappa_2\frac{d^2\log^6H}{H^2} + \kappa_2\calL^{\text{marginal}}(\wh{\theta}).
    \end{split}\end{equation}
    Here $(i)$ is because $c_2^2\beta_1d/(\alpha_1(\alpha_1/\beta_1-c_2)) \leq \kappa_2d^2\log^6H/H^2$ for sufficiently large $c'$. 

    Finally, we connect back to the projected policy $\pi_0^{\wh{\theta}}$.
    By definition, applying the projection $\text{Proj}_{\calA_i}(\bma) = \arg\min_{\bma'\in\calA_i}\norm{\bma'-\bma}$ to samples from $\pi^{\wh{\theta},\noclip}_0(\cdot)$ yields $\pi^{\wh{\theta}}_0(\cdot)$. Moreover, applying $\text{Proj}_{\calA_i}(\cdot)$ to $\pi^\star(\cdot)$ doesn't change the sample distribution since $\pi^\star(\cdot)$ is only defined on $\calA_i$. 
    Therefore, by the data processing inequality, we have $D_{\tv}(\pi^\star_0,\pi^{\wh{\theta}}_0) \leq D_{\tv}(\pi^\star_0,\pi^{\wh{\theta},\noclip}_0)$. Combining with \Cref{eq:main_4} gives the following for $N \geq 2c_1c_w^2\log(2\kappa_1c_1c_w\sqrt{N}/\delta)/9 \geq 2c_w^2\log(2|\calF_{\tilde{\Theta}}|/\delta)/9$:
    \begin{equation}\begin{split}\label{eq:main_5}
        {}& D_{\tv}^2\b{\pi^\star_0,\pi^{\wh{\theta}}_0} \leq 3\kappa_2\frac{d^2\log^6H}{H^2} + \kappa_2\calL^{\text{marginal}}(\wh{\theta})\\
        \leq {}& 3\kappa_2\frac{d^2\log^6H}{H^2} + \kappa_2\b{4\sqrt{2}\kappa_1\sqrt{c_1\log \frac{2\kappa_1c_1c_w\sqrt{N}}{\delta}} + 4}\frac{\sqrt[4]{\bbE_{\bma\sim\tilp} (w(\bma)^4)}}{\sqrt{N}}\\
        \lesssim {}& \frac{d^2\log^6H}{H^2} + \frac{d\sqrt[4]{\bbE_{\bma\sim\tilp} (w(\bma)^4)}}{\sqrt{N}}\\
        = {}& \frac{d^2\log^6H}{H^2} + \frac{d\exp\b{\frac{3}{4}D_4\b{\pi^{\star}\|\tilp}}}{\sqrt{N}}.
    \end{split}\end{equation}
    Here in the second last line, the dependencies on constants and polylog factors $\kappa_1=3(c'+c'')c_2^2(1+c_b+2\sqrt{d})^2\log H$, $\kappa_2$, $\sqrt{c_1\log ({2\kappa_1c_1c_w\sqrt{N}}/{\delta})}$ are absorbed using $\lesssim$. In the last line, we used the definition of fourth-order Renyi divergence $D_4(p||q)=\frac{1}{3}\ln\bbE_{x\sim q}\bs{(p(x)/q(x))^4}$.
\end{proof}

\subsubsection{Supporting Details}
\begin{lemma}[Covering and boundedness of $\calF_{\Theta}$]\label{lem:covering_of_loss}
    Consider the setting of \Cref{thm:isflow} and notations in \Cref{eq:main_1}. The $\kappa$-covering number of $\calF_{\Theta}=\{f^{\theta}(\cdot)=w(\cdot)l^{\theta}(\cdot): \theta\in\Theta\}$, denoted by $\calC(\calF_{\Theta},\kappa)$, satisfies 
    \begin{align}
        \log\calC(\calF_{\Theta},\kappa) \leq c_1\log\frac{3(c'+c'')c_1c_2c_w (1 + c_b+2\sqrt{d})\log H}{\kappa}.
    \end{align}
\end{lemma}

\begin{proof}
    Consider any $\kappa>0$ and any $\kappa$-cover of $\calS_{\Theta}$, denoted by $\calS_{\tilde{\Theta}}$. Then by definition, for any $\theta\in \Theta$, there exists $\tilde{\theta}\in\tilde{\Theta}$ such that $\max_{\tau,\bma}\|s_{\tau}^{\theta}\b{\bma} - s_{\tau}^{\tilde{\theta}}\b{\bma}\| \leq \kappa$.

    We now bound $\|f^{\theta}(\bma)-f^{\tilde{\theta}}(\bma)\|$ for any $\bma$ in the support of $f$, i.e., $\{\bma:\norm{\bma}\leq c_b\}$, as follows. 
    \begin{equation}\begin{split}\label{eq:main_2}
        {}& \norm{f^{\theta}(\bma)-f^{\tilde{\theta}}(\bma)} = \norm{w(\bma)l^{\theta}(\bma) - w(\bma)l^{\tilde{\theta}}(\bma)}\\
        \leq {}& \left\|c_w \sum_{\tau=1}^H \beta_{\tau} \bbE_{\epsilon_{\tau}\sim \calN(0,I)}\left[\left\|s_{\tau}^{\theta}\b{\sqrt{\overline{\alpha}_{\tau}}\bma+\sqrt{1-\overline{\alpha}_{\tau}}\epsilon_{\tau}}- \frac{\epsilon_{\tau}}{\sqrt{1 - \bar{\alpha}_\tau}}\right\|^2\right.\right.\\
        {}& \left.\left. - \norm{s_{\tau}^{\tilde{\theta}}\b{\sqrt{\overline{\alpha}_{\tau}}\bma+\sqrt{1-\overline{\alpha}_{\tau}}\epsilon_{\tau}} - \frac{\epsilon_{\tau}}{\sqrt{1 - \bar{\alpha}_\tau}}}^2\right]\right\|\\
        = {}& c_w \sum_{\tau=1}^H \beta_{\tau} \bbE_{\epsilon_{\tau}\sim \calN(0,I)} \norm{s_{\tau}^{\theta}\b{\sqrt{\overline{\alpha}_{\tau}}\bma+\sqrt{1-\overline{\alpha}_{\tau}}\epsilon_{\tau}}-s_{\tau}^{\tilde{\theta}}\b{\sqrt{\overline{\alpha}_{\tau}}\bma+\sqrt{1-\overline{\alpha}_{\tau}}\epsilon_{\tau}}} \\
        {}&  \cdot\norm{s_{\tau}^{\theta}\b{\sqrt{\overline{\alpha}_{\tau}}\bma+\sqrt{1-\overline{\alpha}_{\tau}}\epsilon_{\tau}}+s_{\tau}^{\tilde{\theta}}\b{\sqrt{\overline{\alpha}_{\tau}}\bma+\sqrt{1-\overline{\alpha}_{\tau}}\epsilon_{\tau}}-\frac{2\epsilon_{\tau}}{\sqrt{1-\overline{\alpha}_{\tau}}}}\\
        \leq {}& c_w \kappa \cdot \sum_{\tau=1}^H \beta_{\tau} \bbE_{\epsilon_{\tau}\sim \calN(0,I)} \left[\norm{s_{\tau}^{\theta}\b{\sqrt{\overline{\alpha}_{\tau}}\bma+\sqrt{1-\overline{\alpha}_{\tau}}\epsilon_{\tau}}} \right.\\
        {}& \left. + \norm{s_{\tau}^{\tilde{\theta}}\b{\sqrt{\overline{\alpha}_{\tau}}\bma+\sqrt{1-\overline{\alpha}_{\tau}}\epsilon_{\tau}}} + \frac{2\norm{\epsilon_{\tau}}}{\sqrt{1 - \bar{\alpha}_\tau}}\right]\\
        \overset{(i)}{\leq} {}& 2c_w \kappa \cdot \sum_{\tau=1}^H \beta_{\tau}  \b{c_2\b{1 + c_b+\sqrt{d}} + \sqrt{\frac{d}{1-\overline{\alpha}_{\tau}}}}\\
        \leq {}& 2c_2c_w \b{1 + c_b+2\sqrt{d}} \kappa \cdot \sum_{\tau=1}^H \frac{\beta_{\tau}}{1-\overline{\alpha}_{\tau}} \\
        \overset{(ii)}{\leq} {}& 3(c'+c'')c_2c_w \b{1 + c_b+2\sqrt{d}}\log H\cdot \kappa.
    \end{split}\end{equation}
    Here $(i)$ is by \Cref{lem:bounded_l} and $(ii)$ is by \Cref{lem:schedule}. Therefore, $\calF_{\tilde{\Theta}}$ is a $\kappa'$-cover of $\calF_{\Theta}$ with $\kappa'=3(c'+c'')c_2c_w \b{1 + c_b+2\sqrt{d}}\log H\cdot \kappa$. We can then conclude that the $\kappa'$-log covering number of $\calF_{\Theta}$ is upper bounded by:
    \begin{equation}\begin{split}
        \log\calC(\calF_{\Theta},\kappa') \leq \log|\calF_{\tilde{\Theta}}| = c_1\log\frac{c_1}{\kappa} = c_1\log\frac{3(c'+c'')c_1c_2c_w (1 + c_b+2\sqrt{d})\log H}{\kappa'}.
    \end{split}\end{equation}
\end{proof}

\begin{lemma}\label{lem:bounded_l} Suppose $H \geq 3c''\log H$. Consider score function $s_{\tau}^{\theta}(\cdot)$ satisfying \Cref{assmp:score_comp} (iii) with constant $c_2$ and loss $l^{\theta}(\cdot)$ defined in \Cref{eq:main_1}. They satisfy the following for all $\theta$, $\tau\in[H]$, and $\bma$ with $\norm{\bma}\leq c_b$:
\begin{subequations}\begin{align}    
    {}& \bbE_{\epsilon\sim\calN(0,I)}\norm{s_{\tau}^{\theta}\b{\sqrt{\overline{\alpha}_{\tau}}\bma+\sqrt{1-\overline{\alpha}_{\tau}}\epsilon}}^2\leq c_2^2\b{1 + c_b+\sqrt{d}}^2,\label{eq:bounded1}\\
    {}& \norm{l^{\theta}(\bma)} \leq 3(c'+c'') \b{c_2^2\b{1+c_b+\sqrt{d}}^2 + d}\log H.\label{eq:bounded2}
\end{align}
\end{subequations}
\end{lemma}

\begin{proof}
    We first prove \Cref{eq:bounded1} as follows:
    \begin{equation}\begin{split}\label{eq:bounded_1}
        \bbE_{\epsilon\sim\calN(0,I)}\norm{s_{\tau}^{\theta}\b{\sqrt{\overline{\alpha}_{\tau}}\bma+\sqrt{1-\overline{\alpha}_{\tau}}\epsilon}}^2
        \overset{(i)}{\leq} {}& c_2^2\bbE_{\epsilon\sim\calN(0,I)}\bs{\b{1 + \sqrt{\overline{\alpha}_{\tau}}\norm{\bma}+\sqrt{1-\overline{\alpha}_{\tau}}\norm{\epsilon}}^2}\\
        \overset{(ii)}{\leq} {}& c_2^2\b{1 + c_b\sqrt{\overline{\alpha}_{\tau}}+\sqrt{d}\sqrt{1-\overline{\alpha}_{\tau}}}^2\\
        \leq {}& c_2^2\b{1 + c_b+\sqrt{d}}^2.
    \end{split}\end{equation}
    Here $(i)$ is by \Cref{assmp:score_comp} (iii); $(ii)$ is by expanding the square and using $\bbE\norm{\epsilon}\leq\sqrt{\bbE\|{\epsilon}\|^2}=\sqrt{d}$.

    To prove \Cref{eq:bounded2}, recall the definition of $l^{\theta}(\bma)$:
    \begin{equation}\begin{split}
        l^{\theta}(\bma) \coloneqq \sum_{\tau=1}^H \beta_{\tau} \bbE_{\epsilon_{\tau}\sim \calN(0,I)}\norm{s_{\tau}^{\theta}\b{\sqrt{\overline{\alpha}_{\tau}}\bma+\sqrt{1-\overline{\alpha}_{\tau}}\epsilon_{\tau}} - \frac{\epsilon_{\tau}}{\sqrt{1 - \bar{\alpha}_\tau}}}^2.
    \end{split}\end{equation}
    Therefore,
    \begin{equation}\begin{split}
        \norm{l^{\theta}(\bma)} \leq {}& 2\sum_{\tau=1}^H \beta_{\tau} \cdot \bbE_{\epsilon_{\tau}\sim \calN(0,I)}\left[\norm{s_{\tau}^{\theta}\b{\sqrt{\overline{\alpha}_{\tau}}\bma+\sqrt{1-\overline{\alpha}_{\tau}}\epsilon_{\tau}}}^2 + \frac{\norm{\epsilon_{\tau}}^2}{1-\overline{\alpha}_{\tau}}\right]\\
        \overset{(i)}{\leq} {}& 2\sum_{\tau=1}^H \beta_{\tau} \cdot \b{c_2^2\b{1+c_b+\sqrt{d}}^2 + \frac{d}{1-\overline{\alpha}_{\tau}}}\\
        \leq {}& 2 \b{c_2^2\b{1+c_b+\sqrt{d}}^2 + d}\sum_{\tau=1}^H \frac{\beta_{\tau}}{1-\overline{\alpha}_{\tau}}\\
        \overset{(ii)}{\leq} {}& 3(c'+c'') \b{c_2^2\b{1+c_b+\sqrt{d}}^2 + d}\log H.
    \end{split}\end{equation}
    Here $(i)$ is by \Cref{eq:bounded_1} and $(ii)$ is by \Cref{lem:schedule}.
\end{proof}

\begin{lemma}\label{lem:schedule} 
    Suppose $H \geq 3c''\log H$. Then the schedule specified in \Cref{eq:scheduler} satisfies
    \begin{equation}\begin{split}
        \sum_{\tau=1}^H \frac{\beta_{\tau}}{\alpha_{\tau}(1-\overline{\alpha}_{\tau})} \leq \frac{3}{2}(c' + c'')\log H
    \end{split}\end{equation}
\end{lemma}

\begin{proof}
    We first notice that
    \begin{equation}\begin{split}
        \frac{\beta_{\tau}}{\alpha_{\tau}(1-\overline{\alpha}_{\tau})} = {}& \frac{1-\alpha_{\tau}}{\alpha_{\tau}(1-\overline{\alpha}_{\tau})} = -\frac{1-\overline{\alpha}_{\tau-1}}{1-\overline{\alpha}_{\tau}} + \frac{1}{\alpha_{\tau}}
        = 1-\frac{1-\overline{\alpha}_{\tau-1}}{1-\overline{\alpha}_{\tau}} + \frac{1-\alpha_{\tau}}{\alpha_{\tau}}\\
        \leq {}& \log\b{\frac{1-\overline{\alpha}_{\tau}}{1-\overline{\alpha}_{\tau-1}}} + \frac{\beta_{\tau}}{\alpha_{\tau}}, \quad \tau\geq 2.
    \end{split}\end{equation}
    For $\tau=1$, we have $\frac{\beta_{\tau}}{\alpha_{\tau}(1-\overline{\alpha}_{\tau})} 
        = 1-\frac{1-\overline{\alpha}_{\tau-1}}{1-\overline{\alpha}_{\tau}} + \frac{1-\alpha_{\tau}}{\alpha_{\tau}}
        = 1 + \frac{\beta_{\tau}}{\alpha_{\tau}}$.
    Summing up $\tau\in[H]$ gives
    \begin{equation}\begin{split}
        \sum_{\tau=1}^H \frac{\beta_{\tau}}{\alpha_{\tau}(1-\overline{\alpha}_{\tau})} \leq {}& 1 + \sum_{\tau=2}^H \log\b{\frac{1-\overline{\alpha}_{\tau}}{1-\overline{\alpha}_{\tau-1}}} + \sum_{\tau=1}^H\frac{\beta_{\tau}}{\alpha_{\tau}}
        \leq 1 + \log\b{\frac{1-\overline{\alpha}_{H}}{1-\overline{\alpha}_{1}}} + \frac{\sum_{\tau=1}^H \beta_{\tau}}{\min_{\tau}\alpha_{\tau}}\\
        \overset{(i)}{\leq} {}& 1 + c'\log H + \frac{c''\log H}{1-c''\log H/H}
        \overset{(ii)}{\leq} \frac{3}{2}(c' + c'')\log H.
    \end{split}\end{equation}
    Here in $(i)$ we have used $(1-\overline{\alpha}_H)/(1-\overline{\alpha}_1)\leq 1/\beta_1 = H^{c'}$, $\sum_{\tau=1}^H \beta_{\tau} \leq \sum_{\tau=1}^H c''\log H/H \leq c''\log H$, and $\alpha_{\tau}\geq 1-c''\log H/H$, and $(ii)$ follows by $H \geq 3c''\log H$.
\end{proof}

\subsubsection{One-step Denoising Bound}
\begin{lemma}[One-step Denoising KL/TV Bound]\label{lem:one_step_denoising}
    Consider random variable $X_0\in\bbR^d$ with score function $s_0^\star(x_0)=\nabla_{x_0}\log\bbP_{X_0}(x_0)$ and one-step forward process 
    \begin{equation}\begin{split}
        X_1 = \sqrt{\alpha}X_0 + \sqrt{\beta}\epsilon_x,
    \end{split}\end{equation}
    where $\epsilon_x\sim\calN(0, I)$, $\alpha,\beta$ are positive scalars satisfying $\alpha+\beta=1$. Let $s_1^\star(x_1)=\nabla_{x_1}\log\bbP_{X_1}(x_1)$ be the score function of $X_1$.
    Consider random variable $Y_1\in\bbR^d$ and one-step backward process
    \begin{equation}\begin{split}\label{eq:backward}
        Y_0 = \b{Y_1+\beta s_1(Y_1)}/\sqrt{\alpha} + \sigma\epsilon_y,
    \end{split}\end{equation}
    where $\epsilon_y\sim\calN(0, I)$, $\sigma>0$, and $s_1(y_1)$ is the approximate score function of $X_1$. Suppose $s_0^\star(\cdot)$ is $c_2$-Lipschitz and suppose $\alpha/\beta>c_2$. We have:
    \begin{equation}\begin{split}
        D_{\tv}^2(\bbP_{X_0}\|\bbP_{Y_0}) \leq {}& A + 2D_{\tv}^2(\bbP_{X_1}\|\bbP_{Y_1}), \quad D_{\kl}(\bbP_{X_0}\|\bbP_{Y_0}) \leq A + D_{\kl}(\bbP_{X_1}\|\bbP_{Y_1}),
    \end{split}\end{equation}
    where $A = \frac{\beta^2}{2\sigma^2\alpha}\bbE_{x_1\sim X_1}\norm{s_1^\star(x_1)-s_1(x_1)}^2 + \frac{\sigma^2 d}{\alpha/\beta-c_2}(c_2^2 + (\frac{1}{\sigma^2}-\frac{\alpha}{\beta})^2)$.
\end{lemma}


\begin{proof}
    For notational simplicity, we will use $\bbP_{U|v}(u)$ to abbreviate conditional probability $\bbP_{U|V}(u|v)$ in this proof. We first prove the TV bound.

    \textbf{Step 1. Connecting $\bm{D_{\tv}(\bbP_{X_0}\|\bbP_{Y_0})}$ to the score difference.}
    Define auxiliary random variable $\tilX_0 = \b{X_1+\beta s_1(X_1)}/\sqrt{\alpha} + \sigma\epsilon$ for $\epsilon\sim\calN(0,I)$.
    By the triangle inequality, we have that
    \begin{equation}\begin{split}\label{eq:laststep_1}
        D_{\tv}^2(\bbP_{X_0}\|\bbP_{Y_0}) \leq {}& 2D_{\tv}^2(\bbP_{X_0}\|\bbP_{\tilX_0}) + 2D_{\tv}^2(\bbP_{\tilX_0}\|\bbP_{Y_0}).
    \end{split}\end{equation}
    Since $\tilX_0$ and $Y_0$ are generated by the same backward process (\Cref{eq:backward}) from distributions $X_1$ and $Y_1$, respectively, by the data processing inequality, we know that $D_{\tv}(\bbP_{\tilX_0}\|\bbP_{Y_0}) \leq D_{\tv}(\bbP_{X_1}\|\bbP_{Y_1})$. Substituting back into \Cref{eq:laststep_1} gives
    \begin{equation}\begin{split}
        D_{\tv}^2(\bbP_{X_0}\|\bbP_{Y_0}) \leq {}& 2D_{\tv}^2(\bbP_{X_0}\|\bbP_{\tilX_0}) + 2D_{\tv}^2(\bbP_{X_1}\|\bbP_{Y_1}).
    \end{split}\end{equation}
    For the first term, we have that
    \begin{equation}\begin{split}
        D_{\tv}(\bbP_{X_0}\|\bbP_{\tilX_0}) = {}& \frac{1}{2}\int \left|\int \b{\bbP_{X_0|x_1}(x_0) - \bbP_{\tilX_0|x_1}(x_0)}\bbP_{X_1}(x_1)\d x_1\right|\d x_0\\
        \leq {}& \frac{1}{2}\int\int \left|\bbP_{X_0|x_1}(x_0) - \bbP_{\tilX_0|x_1}(x_0)\right|\d x_0 \bbP_{X_1}(x_1) \d x_1\\
        = {}& \bbE_{X_1}\bs{D_{\tv}\b{\bbP_{X_0|x_1}\|\bbP_{\tilX_0|x_1}}}. 
    \end{split}\end{equation}
    Substituting back gives
    \begin{equation}\begin{split}\label{eq:laststep_0}
        D_{\tv}^2(\bbP_{X_0}\|\bbP_{Y_0}) \leq {}& 2\bbE_{X_1}^2\bs{D_{\tv}\b{\bbP_{X_0|x_1}\|\bbP_{\tilX_0|x_1}}} + 2D_{\tv}^2(\bbP_{X_1}\|\bbP_{Y_1})\\
        \leq {}& 2\bbE_{X_1}\bs{D_{\tv}^2\b{\bbP_{X_0|x_1}\|\bbP_{\tilX_0|x_1}}} + 2D_{\tv}^2(\bbP_{X_1}\|\bbP_{Y_1})\\
        \leq {}& \bbE_{X_1}\bs{D_{\kl}\b{\bbP_{X_0|x_1}\|\bbP_{\tilX_0|x_1}}} + 2D_{\tv}^2(\bbP_{X_1}\|\bbP_{Y_1}).
    \end{split}\end{equation}
    Here the last line is by Pinsker's inequality.

    Since $\tilX_0|x_1 \sim \calN((x_1+\beta s_1(x_1))/\sqrt{\alpha}, \sigma^2 I)$, its negative log density $-\log\bbP_{\tilX_0|x_1}$ is $1/\sigma^2$-strongly convex. Therefore, by \cite[Equation (4)]{vempala2019rapid}, 
    \begin{equation}\begin{split}
        D_{\kl}(\bbP_{X_0|x_1}\|\bbP_{\tilX_0|x_1}) \leq {}& \frac{\sigma^2}{2} \bbE_{x_0\sim \bbP_{X_0|x_1}}\norm{\nabla_{x_0}\log \bbP_{X_0|x_1}(x_0) - \nabla_{x_0}\log \bbP_{\tilX_0|x_1}(x_0)}^2
    \end{split}\end{equation}
    By definition of $X_0|x_1$ and $\tilX_0|x_1$, we can write the log probabilities explicitly:
    \begin{equation}\begin{split}\label{eq:laststep_2}
        \nabla_{x_0}\log \bbP_{X_0|x_1}(x_0) = {}& \nabla_{x_0}\log \b{\frac{\bbP_{X_1|x_0}(x_1)\bbP_{X_0}(x_0)}{\bbP_{X_1}(x_1)}}= \nabla_{x_0}\b{\log \bbP_{X_1|x_0}(x_1) + \log\bbP_{X_0}(x_0)}\\
        = {}& - \frac{\alpha}{\beta}\b{x_0-\frac{x_1}{\sqrt{\alpha}}} + s_0^\star(x_0),\\
        \nabla_{x_0}\log \bbP_{\tilX_0|x_1}(x_0) = {}& -\frac{1}{\sigma^2}\b{x_0-\frac{x_1+\beta s_1(x_1)}{\sqrt{\alpha}}}.
    \end{split}\end{equation}
    Substituting back gives
    \begin{equation}\begin{split}\label{eq:1to0_2}
        {}&  D_{\kl}(\bbP_{X_0|x_1}\|\bbP_{\tilX_0|x_1}) \leq \frac{\sigma^2}{2} \bbE_{x_0\sim \bbP_{X_0|x_1}}\norm{\underbrace{s_0^\star(x_0) - \frac{\beta}{\sigma^2\sqrt{\alpha}}s_1(x_1) + \b{\frac{1}{\sigma^2}-\frac{\alpha}{\beta}}\b{x_0-\frac{x_1}{\sqrt{\alpha}}}}_{\calI}}^2.
    \end{split}\end{equation}
    Finally, combining the above inequality with \Cref{eq:laststep_0} gives
    \begin{equation}\begin{split}\label{eq:laststep_term0}
        D_{\tv}^2(\bbP_{X_0}\|\bbP_{Y_0}) \leq {}& \bbE_{X_1}\bs{D_{\kl}\b{\bbP_{X_0|x_1}\|\bbP_{\tilX_0|x_1}}} + 2D_{\tv}^2(\bbP_{X_1}\|\bbP_{Y_1})\\
        \leq {}& \frac{\sigma^2}{2} \bbE_{x_1\sim\bbP_{X_1}}\bbE_{x_0\sim \bbP_{X_0|x_1}}\norm{\calI}^2 + 2D_{\tv}^2(\bbP_{X_1}\|\bbP_{Y_1}).
    \end{split}\end{equation}

    \textbf{Step 2. Upper bounding score difference term $\bm{\calI}$.} Before further derivation, we show $\bbP_{X_0|x_1}(x_0)$ is strongly log-concave. Note that 
    \begin{equation}\begin{split}
        -\nabla_{x_0}^2\log \bbP_{X_0|x_1}(x_0) = {}& - \nabla_{x_0} \b{s_0^\star(x_0) - \frac{\alpha}{\beta}\b{x_0-\frac{x_1}{\sqrt{\alpha}}}} = -\nabla_{x_0} s_0^\star(x_0) + \frac{\alpha}{\beta}I\\
        \succeq {}& \b{\frac{\alpha}{\beta}-c_2}I.
    \end{split}\end{equation}
    Here the last line is because $s_0^\star(\cdot)$ is $c_2$-Lipschitz. Since $\alpha/\beta > c_2$, $\bbP_{X_0|x_1}$ is $(\alpha/\beta-c_2)$-strongly-log-concave. Utilizing this fact, we will now upper bound 
    \begin{equation}\begin{split}\label{eq:laststep_term00}
        \bbE_{x_0\sim \bbP_{X_0|x_1}}\norm{\calI}^2 \leq \|\bbE_{x_0\sim \bbP_{X_0|x_1}}[\calI]\|^2 + \tr\b{\text{Var}_{x_0\sim\bbP_{X_0|x_1}}\bs{\calI}}.        
    \end{split}\end{equation}

    Firstly, we derive $\bbE_{x_0\sim \bbP_{X_0|x_1}}[\calI]$. 
    Since $\bbP_{X_0|x_1}$ is strongly log-concave, it decays exponentially to $0$ as $\norm{x_0}\to\infty$. Therefore, by the divergence theorem, for all $i\in[d]$ 
    \begin{equation}\begin{split}
        \int\text{div}\b{\bbP_{X_0|x_1}(x_0)e_i}\d x_0 = \oint_{\infty} \bbP_{X_0|x_1}(x_0)e_i\cdot \hat{n}\d S = 0.
    \end{split}\end{equation}
    Therefore,
    \begin{equation}\begin{split}
        0 = {}& \int\text{div}\b{\bbP_{X_0|x_1}(x_0)e_i}\d x_0 = \int \nabla_{x_0} \bbP_{X_0|x_1}(x_0)\cdot e_i \d x_0 + \int \bbP_{X_0|x_1}(x_0) \text{div}(e_i) \d x_0\\
        = {}& \int \nabla \bbP_{X_0|x_1}(x_0)\cdot e_i \d x_0.
    \end{split}\end{equation}
    Since the above holds for all $i\in[d]$, we know that
    \begin{equation}\begin{split}
        0 = {}& \int \nabla_{x_0} \bbP_{X_0|x_1}(x_0) \d x_0 = \bbE_{X_0|x_1} \bs{\nabla_{x_0} \log\bbP_{X_0|x_1}(x_0)}\\
        = {}& \bbE_{X_0|x_1}\bs{- \frac{\alpha}{\beta}\b{x_0-\frac{x_1}{\sqrt{\alpha}}} + s_0^\star(x_0)}.
    \end{split}\end{equation}
    Here the last line is by \Cref{eq:laststep_2}.
    Therefore, combining with the definition of $\calI$ gives
    \begin{equation}\begin{split}\label{eq:laststep_term1}
        \bbE_{x_0\sim \bbP_{X_0|x_1}}\bs{\calI} = {}& \bbE_{x_0\sim \bbP_{X_0|x_1}}\bs{\frac{1}{\sigma^2}\b{x_0-\frac{x_1}{\sqrt{\alpha}}} - \frac{\beta}{\sigma^2\sqrt{\alpha}}s_1(x_1)}\\
        = {}& \frac{1}{\sigma^2}\b{\frac{x_1+\beta s_1^\star(x_1)}{\sqrt{\alpha}}-\frac{x_1}{\sqrt{\alpha}}} - \frac{\beta}{\sigma^2\sqrt{\alpha}}s_1(x_1)\\
        = {}& \frac{\beta}{\sigma^2\sqrt{\alpha}}\b{s_1^\star(x_1) - s_1(x_1)}.
    \end{split}\end{equation}
    Here the second line is by the Tweedie's identity.

    Secondly, we bound $\tr(\text{Var}_{x_0\sim\bbP_{X_0|x_1}}\bs{\calI})$. By the definition of $\calI$, 
    \begin{equation}\begin{split}\label{eq:laststep_term2}
        \tr(\text{Var}_{x_0\sim\bbP_{X_0|x_1}}\bs{\calI}) = {}& \sum_{i=1}^d\text{Var}_{x_0\sim\bbP_{X_0|x_1}}\bs{s_0^\star(x_0) + \b{\frac{1}{\sigma^2}-\frac{\alpha}{\beta}}x_0}_i.
    \end{split}\end{equation}
    Since $\bbP_{X_0|x_1}$ is $(\alpha/\beta-c_2)$-strongly-log-concave, it satisfies the Poincare inequality \cite[Equation (24)]{vempala2019rapid} with constant $\frac{1}{\alpha/\beta-c_2}$, leading to:
    \begin{equation}\begin{split}
       \tr(\text{Var}_{x_0\sim\bbP_{X_0|x_1}}\bs{\calI}) = {}& \sum_{i=1}^d\text{Var}_{x_0\sim\bbP_{X_0|x_1}}\bs{s_0^\star(x_0) + \b{\frac{1}{\sigma^2}-\frac{\alpha}{\beta}}x_0}_i\\
        \leq {}& \frac{1}{\alpha/\beta-c_2}\sum_{i=1}^d\bbE_{x_0\sim\bbP_{X_0|x_1}}\bs{\norm{\nabla \bs{s_0^\star(x_0) + \b{\frac{1}{\sigma^2}-\frac{\alpha}{\beta}}x_0}_i }^2}\\
        \leq {}& \frac{2d}{\alpha/\beta-c_2}\b{c_2^2 + \b{\frac{1}{\sigma^2}-\frac{\alpha}{\beta}}^2}.
    \end{split}\end{equation}

    Combining \Cref{eq:laststep_term00,eq:laststep_term1,eq:laststep_term2} gives
    \begin{equation}\begin{split}
        \bbE_{x_0\sim \bbP_{X_0|x_1}}\norm{\calI}^2 \leq \frac{\beta^2}{\sigma^4\alpha}\norm{s_1^\star(x_1)-s_1(x_1)}^2 + \frac{2\beta d}{\alpha-c_2\beta}\b{c_2^2 + \b{\frac{1}{\sigma^2}-\frac{\alpha}{\beta}}^2}.
    \end{split}\end{equation}

    \textbf{Step 3. Conclusion. } Substituting the above inequality back into \Cref{eq:laststep_term0} gives
    \begin{equation}\begin{split}
        D_{\tv}^2(\bbP_{X_0}\|\bbP_{Y_0})\leq {}& \frac{\beta^2}{2\sigma^2\alpha}\bbE_{x_1\sim X_1}\norm{s_1^\star(x_1)-s_1(x_1)}^2\\
        {}&  + \frac{\sigma^2\beta d}{\alpha-c_2\beta}\b{c_2^2 + \b{\frac{1}{\sigma^2}-\frac{\alpha}{\beta}}^2} + 2D_{\tv}^2(\bbP_{X_1}\|\bbP_{Y_1}).
    \end{split}\end{equation}

    Similarly, one can derive a bound on KL divergence:
    \begin{equation}\begin{split}
        {}&  D_{\kl}(\bbP_{X_0}\|\bbP_{Y_0}) \leq \bbE_{X_1}\bs{D_{\kl}\b{\bbP_{X_0|x_1}\|\bbP_{\tilX_0|x_1}}} + D_{\kl}(\bbP_{X_1}\|\bbP_{Y_1})\\
        \leq {}& \frac{\sigma^2}{2} \bbE_{x_1\sim\bbP_{X_1}}\bbE_{x_0\sim \bbP_{X_0|x_1}}\norm{\calI}^2 + D_{\kl}(\bbP_{X_1}\|\bbP_{Y_1})\\
        \leq {}& \frac{\beta^2}{2\sigma^2\alpha}\bbE_{x_1\sim X_1}\norm{s_1^\star(x_1)-s_1(x_1)}^2 + \frac{\sigma^2\beta d}{\alpha-c_2\beta}\b{c_2^2 + \b{\frac{1}{\sigma^2}-\frac{\alpha}{\beta}}^2} + D_{\kl}(\bbP_{X_1}\|\bbP_{Y_1}).
    \end{split}\end{equation}
\end{proof}

\section{Proof of \Cref{thm:isflow_negap}}\label{sec:proof-end-to-end}
\begin{theorem}[Theorem 5 in \citet{marl_theory1}]\label{thm:marl}
    Suppose state space $\calS$ and action spaces $\{\calA_i\}_{i\in[n]}$ are finite. Suppose all stationary points of the return function $J(\pi)$ are isolated. Suppose the estimated averaged $Q$ functions $\wh{Q}_i^{\pi^k}$ are accurate, i.e., $\wh{Q}_i^{\pi^k} = \overline{Q}_i^{\pi^k}$. Suppose the updated policies $\pi_i^k =\pi_i^{\star,k}$. Then for learning rate $\eta=(1-\gamma)^2/(2n(r_{\max}-r_{\min}))$, 
    \begin{equation*}\begin{split}
        \frac{\sum_{k=0}^{K-1} \negap\b{\pi^{k}}^2}{K} \leq \frac{6M(r_{\max}-r_{\min})^2}{c(1-\gamma)^3} \frac{n}{K}.
    \end{split}\end{equation*}

    Here related constants are defined as follows:
    \begin{equation*}\begin{split}
        r_{\max} \coloneqq {}& \max_{\bms,\bma}r(\bms,\bma), \quad r_{\min} = \min_{\bms,\bma}r(\bms,\bma),\quad
        M\coloneqq \sup_{\pi} \max_{\bms} \frac{1}{d^{\pi}(\bms)}, \quad c = \inf_{k} c_{\pi^k},\\
        d^{\pi}(\bms) \coloneqq {}& (1-\gamma)\bbE_{\bms_0\sim\rho^0} \bs{\sum_{t=0}^\infty \gamma^t\bbP\b{\bms_t=\bms|\bms_0}},
        \quad c_{\pi}\coloneqq \min_{\bms}\sum_{\bma_i^\star\in\arg\max_{\bma_i}\overline{Q}_i^{\pi}\b{\bms,\bma_i}} \pi_i\b{\bma_i^\star|\bms}.
    \end{split}\end{equation*}
\end{theorem}
Intuitively, the quadratic mean of the \negap s converges to $0$ at a rate of $\calO(\sqrt{n/K})$, where $n$ is the number of agents and $K$ is the number of epochs. The output policy is therefore an approximate NE.

\begin{theorem}\label{thm:isflow_negap}
    Consider the setting of \Cref{thm:marl} and \Cref{alg:exact}. Suppose in every epoch $k$, the learned policy $\pi^{k}_i$ of all agents $i\in[n]$ satisfies $D_{\tv}(\pi_i^{k}(\cdot|\bms), \pi_i^{\star,k}(\cdot|\bms)) \leq \epsilon_{\is}$ for all $i\in[n], k\in[K], \bms\in\calS$, then
    \begin{equation*}\begin{split}
        \frac{\sum_{k=0}^{K-1} \negap\b{\pi^{k}}^2}{K}\leq \frac{6M(r_{\max}-r_{\min})^2}{c(1-\gamma)^3}\frac{n}{K} + \frac{12M(r_{\max}-r_{\min})^2}{c(1-\gamma)^4}n^2\epsilon_{\is},
    \end{split}\end{equation*}
    where $M$ and $c$ are defined in \Cref{thm:marl}. \qed
\end{theorem}

\begin{proof}[Proof of \Cref{thm:isflow_negap}]
    We let $\pi^k_i(\cdot|\bms)$ denote the approximate policy learned by importance sampling score matching with $N$ samples from $\pi^{k-1}_i(\cdot|\bms)$, and let $\tilde{\pi}^k_i(\cdot|\bms)$ denote the exact softmax policy update from $\pi^{k-1}$, i.e., $\tilde{\pi}_i^k(\cdot|\bms) \propto \pi_i^{k-1}(\cdot|\bms)\exp\b{\frac{\eta}{1-\gamma}\overline{Q}_i^{\pi^{k-1}}(\bms,\cdot)}$. By assumption, $D_{\tv}\b{\tilde{\pi}_i^k(\cdot|\bms), \pi_i^k(\cdot|\bms)} \leq \epsilon_{\is}$ for all $i\in[n], k\in[K], \bms\in\calS$.

    \textbf{Step 1: Perturbation bound on $J$.}
    We first bound $\norm{J(\pi)-J(\tilde{\pi})}$ for any two joint policies $\pi=(\pi_1,\cdots,\pi_n)$ and $\tilde{\pi}=(\tilde{\pi}_1,\cdots,\tilde{\pi}_n)$ satisfying $D_{\tv}(\pi_i(\cdot|\bms), \tilde{\pi}_i(\cdot|\bms)) \leq \epsilon$ for all $i, \bms$. To do this, we introduce auxiliary policies $\nu^i \coloneqq \left(\pi_1, \ldots, \pi_i,\, \tilde{\pi}_{i+1}, \ldots, \tilde{\pi}_n\right), \forall i\in[n]$ and $\nu^0 \coloneqq \tilde{\pi}$.
    By telescoping, we know that
    \begin{equation}\label{eq:telescope-agents}
        J(\pi) - J(\tilde{\pi}) = \sum_{i=1}^{n} \b{J(\nu^i) - J(\nu^{i-1})}.
    \end{equation}
    Since intermediate policies $\nu^i$ and $\nu^{i-1}$ only differ in the $i$-th component, we can apply the performance difference lemma (Lemma~8 of \cite{marl_theory1}) and get:
    \begin{equation}\label{eq:pdl-agent-i}
        J(\nu^i) - J(\nu^{i-1})
        = \frac{1}{1-\gamma} \sum_{\bms,\bma_i} d^{\nu^i}(\bms) \pi_i(\bma_i|\bms)\, \overline{A}^{\nu^{i-1}}_{i}(\bms, \bma_i).
    \end{equation}
    Here $\overline{A}^{\pi}_i(\bms, \bma_i) := \sum_{\bma_{-i}} \pi_{-i}(\bma_{-i}|\bms)\, A^{\pi}_i(\bms, \bma_i, \bma_{-i})$ and $A^{\pi}_i(\bms, \bma_i, \bma_{-i})=Q^{\pi}_i(\bms, \bma_i, \bma_{-i})-V^{\pi}_i(\bms)$.
    By the definition of $\overline{A}^{\pi}_i(\bms,\bma_i)$, we know that $\sum_{\bma_i} \pi_i(\bma_i|\bms)\, \overline{A}^{\pi}_i(\bms, \bma_i) = 0$. Since $\nu^{i-1}$ has $\tilde{\pi}_i$ as the $i$-th component, $\sum_{\bma_i} \tilde{\pi}_i(\bma_i|\bms)\, \overline{A}^{\nu^{i-1}}_i(\bms, \bma_i) = 0$. Substituting this equation into \Cref{eq:pdl-agent-i} leads to 
    \begin{equation}
        J(\nu^i) - J(\nu^{i-1})
        = \frac{1}{1-\gamma} \sum_{\bms, \bma_i} d^{\nu^i}(\bms)  \left(\pi_i(\bma_i|\bms) - \tilde{\pi}_i(\bma_i|\bms)\right) \overline{A}^{\nu^{i-1}}_{i}(\bms, \bma_i)
    \end{equation}
    Taking the absolute values gives:
    \begin{align}\label{eq:holder}
        \norm{J(\nu^i) - J(\nu^{i-1})}
        \leq {}& \frac{1}{1-\gamma} \sum_{\bms} d^{\nu^i}(\bms)\,
        \left\|\pi_i(\cdot|\bms) - \tilde{\pi}_i(\cdot|\bms)\right\|_1
        \cdot \max_{\bma_i} \left|\overline{A}^{\nu^{i-1}}_{i}(\bms, \bma_i)\right|\\
        \leq {}& \frac{1}{1-\gamma} \cdot \sum_{\bms} d^{\nu^i}(\bms) 2D_{\tv}\left(\pi_i(\cdot|\bms), \tilde{\pi}_i(\cdot|\bms)\right) \cdot \frac{r_{\max} - r_{\min}}{1 - \gamma}\\
        \leq {}& \frac{2(r_{\max} - r_{\min})}{(1-\gamma)^2}\epsilon.
    \end{align}
    Substituting into \eqref{eq:telescope-agents} gives
    \begin{equation}\label{eq:perturbation_phi}
        \norm{J(\pi) - J(\tilde{\pi})}
        \le \frac{2n(r_{\max} - r_{\min})}{(1-\gamma)^2}\epsilon.
    \end{equation}

    \textbf{Step 2: Reward improvement from importance sampling score matching.}
    By Lemma 20 and Lemma 21 in \citet{marl_theory1}, since $\tilde{\pi}^k$ is the exact softmax update from $\pi^{k-1}$ with accurate Q functions, we have
    \begin{equation}\begin{split}
        J(\tilde{\pi}^k) - J(\pi^{k-1}) \overset{\text{Lemma 20}}{\geq} {}& \frac{1}{\eta}\sum_{i=1}^n\sum_{\bms} d^{\tilde{\pi}^k}(\bms) \log\b{\sum_{\bma_i} \pi_i^{k-1}(\bma_i|\bms) \exp\b{\frac{\eta}{1-\gamma}\overline{A}^{\pi^{k-1}}_i(\bms, \bma_i)}}\\
        \overset{\text{Lemma 21}}{\geq} {}& \frac{c\eta}{3M}\negap\b{\pi^{k-1}}^2,
    \end{split}\end{equation}
    where constants $c$ and $M$ are defined in \Cref{thm:marl}. 
    Combining with Step 1 gives:
    \begin{equation}\begin{split}
        J(\pi^k) - J(\pi^{k-1}) = {}& J(\pi^k) - J(\tilde{\pi}^k) + J(\tilde{\pi}^k) - J(\pi^{k-1})\\
        \geq {}& \frac{c\eta}{3M}\negap\b{\pi^{k-1}}^2 - \frac{2n(r_{\max} - r_{\min})}{(1-\gamma)^2}\epsilon_{\is}.
    \end{split}\end{equation}

    \textbf{Step 3: Connecting reward improvement to Nash gap.}
    Summing from $k=1$ to $K$ gives:
    \begin{equation}\begin{split}
        J(\pi^K) - J(\pi^0) \geq \frac{c\eta}{3M}\sum_{k=0}^{K-1}\negap\b{\pi^{k}}^2 - \frac{2nK(r_{\max} - r_{\min})}{(1-\gamma)^2}\epsilon_{\is}.
    \end{split}\end{equation}
    Since $J(\pi^K) - J(\pi^0) \leq (r_{\max}-r_{\min})/(1-\gamma)$, we have:
    \begin{equation}\begin{split}
        \frac{1}{K}\sum_{k=0}^{K-1}\negap\b{\pi^{k}}^2 \leq {}& \frac{3M(r_{\max}-r_{\min})}{c\eta(1-\gamma)}\frac{1}{K} + \frac{6Mn(r_{\max}-r_{\min})}{c\eta(1-\gamma)^2}\epsilon_{\is}.
    \end{split}\end{equation}
    Substituting $\eta = (1-\gamma)^2/\b{2n(r_{\max}-r_{\min})}$ gives:
    \begin{equation}\begin{split}
        \frac{1}{K}\sum_{k=0}^{K-1}\negap\b{\pi^{k}}^2 \leq {}& \frac{6Mn(r_{\max}-r_{\min})^2}{c(1-\gamma)^3}\frac{1}{K} + \frac{12Mn^2(r_{\max}-r_{\min})^2}{c(1-\gamma)^4}\epsilon_{\is}.
    \end{split}\end{equation}
    This completes the proof.
\end{proof}

\section{Experiment Details}\label{sec:appendix_3}
\begin{table}[h]
    \centering
    \vspace{10pt}
    \begin{tabular*}{\textwidth}{c @{\extracolsep{\fill}} ccc}
    \toprule
      \textbf{Name}   & \textbf{Value}&\textbf{Name}   & \textbf{Value} \\
      \midrule
      Rollout batch size $B_r$ & 20 & Diffusion noise schedule & linear\\
      Update batch size $B$ & 256 & Diffusion noise schedule start & 0.001\\
      Smoothing parameter $\xi$ & 0.005 & Diffusion noise schedule end & 0.999\\
      Score/Q network (MLP) size & 3x256 & Diffusion steps & 20\\
      Score/Q network activation & mish & Adam learning rate & linear(1e-4,5e-5) \\
      Diffusion policy update lr $\eta$ & 1\\
         \bottomrule
    \end{tabular*}
    \caption{Hyperparameters for \Cref{sec:sim}}
    \label{tab:hyper}
\end{table}




\end{document}